%% file: sbct1.tex
\documentclass[11pt]{article}
\usepackage{amsmath,amssymb,amsfonts,amsthm,epsfig}
\usepackage[usenames,dvipsnames]{xcolor}

\usepackage{bm,xspace}

\usepackage{cancel}
\usepackage{fullpage}
\usepackage{stylefile}
\usepackage{framed}
\usepackage{verbatim}
\usepackage{enumitem}
\usepackage{array}
\usepackage{multirow}
\usepackage{afterpage}
\usepackage{mathrsfs}

\usepackage{dsfont}
\usepackage[normalem]{ulem}

\makeatletter
\newtheorem*{rep@theorem}{\rep@title}
\newcommand{\newreptheorem}[2]{
\newenvironment{rep#1}[1]{
 \def\rep@title{#2 \ref{##1}}
 \begin{rep@theorem}\itshape}
 {\end{rep@theorem}}}
\makeatother

 \DeclareMathOperator*{\Ppr}{\textbf{Pr}}
\newreptheorem{theorem}{Theorem}

\def\colorful{0}

\ifnum\colorful=1

\newcommand{\blue}[1]{{{\color{blue}#1}}}
\newcommand{\red}[1]{{\color{red} {#1}}}

\fi
\ifnum\colorful=0

\newcommand{\blue}[1]{{{#1}}}
\newcommand{\red}[1]{{{#1}}}

\fi

\usepackage{tikz}
\usepackage{pgfplots}

\tikzstyle{v}=[circle,fill=black,thick,inner sep=1pt]
\def\Dyes{\mathcal{D}_{\text{yes}}}
\def\Dno{\mathcal{D}_{\text{no}}}

\def\py{p_{\text{yes}}}
\def\pn{p_{\text{no}}}
\def\S{\Eno^*}

\def\Alg{\mathrm{Alg}}
\def\br{\boldsymbol{r}}

\def\eps{\varepsilon}
\def\Ball{\mathrm{Ball}}
\def\Cube{\mathrm{Cube}}
\def\CubeSet{\mathrm{CubeSet}}

\def\convex{\mathrm{convex}}
\def\targetset{S}
\def\cover{\text{cover}}

\newcommand{\e}{\epsilon}

\newcommand{\inv}{^{-1}}

\newcommand{\del}{\partial }

\def\SS{\boldsymbol{S}}

\def\gau{{\mathcal{N}(0,1)^n}}
\def\Dyes{\mathcal{D}_{\textsf{yes}}}
\def\Dno{\mathcal{D}_{\textsf{no}}}
\def\dtv{d_{\text{TV}}}
\def\Eyes{\mathcal{E}_{\textsf{yes}}}
\def\Eno{\mathcal{E}_{\textsf{no}}}
\def\barbx{\mathbf{x}}
\def\bb{\mathbf{b}}
\def\dd{\mathbf{d}}
\def\barby{\mathbf{y}}
\def\by{\mathbf{y}}
\def\br{\mathbf{r}}
\def\bN{\mathbf{N}}
\def\gaus{{\mathcal{N}^n}}
\def\zz{\mathbf{z}}

\title{Sample-based high-dimensional convexity testing}
\author{
Xi Chen\thanks{Columbia University, email: \texttt{xichen@cs.columbia.edu}.}
\and
Adam Freilich\thanks{Columbia University, email: \texttt{freilich@cs.columbia.edu}.}
\and
Rocco A. Servedio\thanks{Columbia University, email: \texttt{rocco@cs.columbia.edu}.}
\and
Timothy Sun\thanks{Columbia University, email: \texttt{tim@cs.columbia.edu}.}
}

\begin{document}

\def\Ayes{\mathcal{A}_{\text{yes}}}
\def\Ano{\mathcal{A}_{\text{no}}}
\def\py{p}
\def\pn{q}
\def\Alg{\textsc{Alg}}
\def\Vyes{\mathcal{V}_{\text{yes}}}
\def\Vno{\mathcal{V}_{\text{no}}}
\def\Byes{\mathcal{B}_{\text{yes}}}
\def\Bno{\mathcal{B}_{\text{no}}}

\maketitle

\begin{abstract}

In the problem of \emph{high-dimensional convexity testing}, there is an unknown set $\targetset \subseteq \R^n$ which is promised to be either convex or $\eps$-far from every convex body with respect to the standard multivariate normal distribution $\normal^n$.  The job of a testing algorithm is then to distinguish between these two cases while making as few inspections of the set $\targetset$ as possible.

In this work we consider \emph{sample-based} testing algorithms, in which the testing algorithm only has access to labeled samples $(\bx,\targetset(\bx))$ where each $\bx$ is independently drawn from $\normal^n$.  We give nearly matching sample complexity upper and lower bounds for both one-sided and two-sided convexity testing algorithms in this framework.  For constant $\eps$, our results show that the sample complexity of one-sided convexity testing is $2^{\tilde{\Theta}(n)}$ samples, while for two-sided convexity testing it is $2^{\tilde{\Theta}(\sqrt{n})}$.

\end{abstract}

\thispagestyle{empty}

\newpage

\setcounter{page}{1}

\input{intro}

\input{prelims}

\input{structural}

\input{1sub}

\input{2slb}

\input{1slb}

\input{2sub}

\begin{flushleft}
\bibliographystyle{alpha}
\bibliography{allrefs}
\end{flushleft}

\appendix

\input{lemmas}

\end{document}

%% file: intro.tex

\section{Introduction}

Over the past few decades the field of property testing has developed into a fertile area with many different branches of active research.  Several distinct lines of work have studied the testability of various kinds of \emph{high-dimensional} objects, including probability distributions (see e.g. \cite{BKR:04long,RubinfeldServedio:05,AAK+07,RX10,ACS10,BFRV11,AcharyaDK15}), Boolean functions (see e.g. \cite{BLR93,PRS02,Blaisstoc09,MORS:10,KMS15} and many other works), and various types of codes and algebraic objects (see e.g. \cite{AKKLRtit,GoldreichSudan06,KaufmanSudan08,BKSSZ10} and many other works).  These efforts have collectively yielded significant insight into the abilities and limitations of efficient testing algorithms for such high-dimensional objects.  A distinct line of work has focused on testing (mostly low-dimensional) \emph{geometric properties}.  Here too a considerable body of work has led to a good understanding of the testability of various low-dimensional geometric properties, see e.g. \cite{CSZ:00,CS:01,Raskhodnikova:03,BMR16icalp,BMR16socg,BMR16fsttcs}.

This paper is about a topic which lies at the intersection of the two general strands (high-dimensional property testing and geometric property testing) mentioned above:  we study the problem of \emph{high-dimensional convexity testing.}  Convexity is a fundamental property which is intensively studied in high-dimensional geometry (see e.g. \cite{convex-geometry,Ball:intro-convex,Szarek06} and many other references) and has been studied in the property testing of images (the two-dimensional case) \cite{Raskhodnikova:03,BMR16icalp,BMR16socg,BMR16fsttcs}, but as we discuss in Section~\ref{sec:relatedwork} below, very little is known about high-dimensional convexity testing.

We consider $\R^n$ endowed with the standard normal distribution $\normal^n$ as our underlying space, so the distance $\dist(S,C)$ between two subsets $S,C \subseteq \R^n$ is $\Pr_{\bx \leftarrow \normal^n}[\bx \in S \bigtriangleup C]$, where   $S\bigtriangleup C$ denotes their symmetric difference.    The standard normal distribution is arguably one of the most natural, and certainly one of the most studied, distributions on $\R^n$.  Several previous works have studied property testing over $\R^n$ with respect to the standard normal distribution, such as the work on testing halfspaces of \cite{MORS:10,BBBY12} and the work on testing surface area of \cite{KNOW:14,Neeman:14}.\vspace{-0.1cm}

\subsection{Our results}

In this paper we focus on \emph{sample-based} testing algorithms for convexity.  Such an algorithm has access to independent draws $(\bx,S(\bx)) \in \R^n \times \{0,1\},$ where $\bx$ is drawn from $\normal^n$ and $S \subseteq \R^n$ is the unknown set being tested for convexity (so in particular the algorithm cannot select points to be queried) with $S(\bx)=1$ if $\bx\in S$.
We say such an algorithm is an \emph{$\eps$-tester for convexity} if it accepts $S$ with probability at least $2/3$
  when $S$ is convex and rejects with probability at least $2/3$ when it is $\eps$-far from convex,
  i.e., $\dist(S,C)\ge \eps$ for all convex sets $C\subseteq \R^n$.
The model of sample-based testing was originally introduced
by Goldreich, Goldwasser, and Ron almost two decades ago \cite{GGR98}, where it was referred to as ``passive testing;'' it has received significant attention over the years \cite{KR00,GGLRS,BBBY12,GoldreichRon16}, with an uptick in research activity in this model over just the past year or so \cite{AHW16, BlaisYoshida16,BMR16icalp,BMR16socg,BMR16fsttcs}.

We consider sample-based testers for convexity that are allowed both one-sided (i.e., the
  algorithm always accepts $S$ when it is convex) and two-sided error. In each case, for constant $\eps>0$ we give nearly matching upper and lower bounds on sample complexity.  Our
  results are as follows:

\begin{theorem} [One-sided lower bound] \label{thm:1slb}
\red{Any one-sided sample-based algorithm that is an $\eps$-tester for convexity over $\normal^n$
  for some $\eps<1/2$
  must use $2^{\Omega(n)}$ samples.}
\end{theorem}

\begin{theorem} [One-sided upper bound] \label{thm:1sub}
For any $\eps > 0$, there is a one-sided sample-based $\eps$-tester for convexity over
  $\normal^n$ which uses $(n/\eps)^{O(n)}$ samples.
\end{theorem}

\begin{theorem} [Two-sided lower bound] \label{thm:2slb}
\red{There exists a positive constant $\eps_0$ such that any
  two-sided sample-based algorithm that is an $\eps$-tester for convexity over $\normal^n$
  for some $\eps\le \eps_0$ must use $2^{\Omega(\sqrt{n})}$ samples.}
\end{theorem}

\begin{theorem} [Two-sided upper bound] \label{thm:2sub}
For any $\eps > 0$, there is a two-sided sample-based $\eps$-tester for convexity over
   $\normal^n$ which uses $n^{O(\sqrt{n}/\eps^2)}$ samples.
\end{theorem}

We will prove Theorems~\ref{thm:1slb}, \ref{thm:1sub}, \ref{thm:2slb} and \ref{thm:2sub} in Sections
\ref{sec:1slb}, \ref{sec:1sub}, \ref{sec:2slb} and~\ref{sec:2sub} respectively.  These results are summarized above in Table~\ref{table:F2}.

\begin{table}[t!]
\renewcommand{\arraystretch}{1.6}
\centering
\begin{tabular}{|m{5.5em}|m{13em}|m{5.7em}|}
\hline
Model & Sample complexity bound & Reference \\ \hline \hline
One-sided & $2^{\Omega(n)}$ samples (for $\eps < 1/2$)  & Theorem~\ref{thm:1slb} \\ \hline
          & $2^{O(n\log(n/\eps))}$ samples & Theorem~\ref{thm:1sub} \\ \hline \hline
Two-sided & $2^{\Omega(\sqrt{n})}$ samples (for $\eps < \eps_0$)  & Theorem~\ref{thm:2slb} \\ \hline
          & $2^{O({\sqrt{n}} \log(n)/\eps^2)}$ samples & Theorem~\ref{thm:2sub} \\ \hline
 \end{tabular}\vspace{0.2cm}
 \caption{Sample complexity bounds for sample-based convexity testing.  In line four,
   ${ \eps_0}>0$ is some absolute constant.}
 \label{table:F2}
 \end{table}

\subsection{Related work} \label{sec:relatedwork}

\noindent {\bf Convexity testing.}
As mentioned above, \cite{Raskhodnikova:03,BMR16fsttcs,BMR16socg,BMR16icalp} studied the testing
  of $2$-dimensional convexity under the uniform distribution, either
  within a compact body such as $[0,1]^2$ \cite{BMR16fsttcs,BMR16socg} or over
  a discrete grid $[n]^2$ \cite{Raskhodnikova:03,BMR16icalp}.
The model of \cite{BMR16fsttcs,BMR16socg} is more closely related to ours:
\cite{BMR16socg} showed that $\Theta(\eps^{-4/3})$
  samples are necessary and sufficient for one-sided sample-based testers, while
\cite{BMR16fsttcs} gave a one-sided general tester (which can make adaptive queries to
  the unknown set) for $2$-dimensional convexity
  with only $O(1/\eps)$ queries.

The only prior work that we are aware of that deals with testing high-dimensional convexity is that of \cite{Vempala}. However, the model
considered in \cite{Vempala} is different from ours in the following important aspects.
First, the goal of an algorithm in their model is to determine whether an unknown $S\subseteq \R^n$ is not convex
  or is $\eps$-close to convex in the following sense: the (Euclidean) volume of $S\bigtriangleup C$, for some convex $C$, is at most
  an $\eps$-fraction of the volume of $S$. Second, in their model an algorithm both can make membership
  queries (to determine whether a given point $x$ belongs to $S$), and can receive
  samples which are guaranteed to be drawn independently and uniformly at random from $S$.
The main result of \cite{Vempala} is an algorithm which uses $(cn/\eps)^n$ many random samples \emph{drawn from $S$},
  for some constant $c$,
  and $\text{poly}(n)/\eps$ membership queries. \vspace{-0.2cm}

\paragraph{Sample-based testing.} A wide range of papers have studied sample-based testing from several different perspectives, including the recent works \cite{BMR16icalp,BMR16socg,BMR16fsttcs} which study sample-based testing of convexity over two-dimensional domains.  In earlier work on sample-based testing, \cite{BBBY12} showed that the class of linear threshold functions can be tested to constant accuracy under $\normal^n$ with $\tilde{\Theta}(n^{1/2})$ samples drawn from $\normal^n$.  (Note that a linear threshold function is a convex set of a very simple sort, as every convex set can be expressed as an intersection of (potentially infinitely many) linear threshold functions.)  The work \cite{BBBY12} in fact gave a characterization of the sample complexity of (two-sided) sample-based testing, in terms of a combinatorial/probabilistic quantity called the ``passive testing dimension.''  This is a distribution-dependent quantity whose definition involves both the class being tested and the distribution from which samples are obtained; it is not \emph{a priori} clear what the value of this quantity is for the class of convex subsets of $\R^n$ and the standard normal distribution $\normal^n$.  Our upper and lower bounds (Theorems~\ref{thm:2sub} and~\ref{thm:2slb}) may be interpreted as giving bounds on the passive testing dimension of the class of convex sets in $\R^n$ with respect to the $\normal^n$ distribution.

\subsection{Our techniques}

\paragraph{One-sided lower bound.}  Our one-sided lower bound has a simple proof using only elementary geometric and probabilistic arguments.  It follows from the fact (see Lemma~\ref{lem:rand-all-on-hull}) that if $q=2^{\Theta(n)}$ many points are drawn independently from $\normal^n$, then with probability $1-o(1)$ no one of the points lies in the convex hull of the $q-1$ others.  This can easily be shown to imply that more than $q$ samples are required (since given only $q$ samples, with probability $1-o(1)$ there is a convex set consistent with any labeling and thus a one-sided algorithm cannot reject).\vspace{-0.2cm}

\blue{
\paragraph{Two-sided lower bound.}  At a high-level, the proof of our two-sided lower bound
   uses the~following standard approach. We first define two distributions $\Dyes$ and $\Dno$
  over sets in $\R^n$ such that
  (i) $\Dyes$ is a distribution over convex sets only, and (ii)
  $\Dno$ is a distribution such that
  $\SS\leftarrow \Dno$ is
  $\eps_0$-far from convex   with probability at least $1-o(1)$ for some positive constant $\eps_0$.
We then show that every sample-based, $q$-query algorithm $A$ with $q=2^{0.01n}$
  must have
\begin{equation}\label{hehe200}
\Ppr_{\SS\leftarrow \Dyes;\hspace{0.05cm}\bx} \big[\text{$A $ accepts $(\bx,\SS(\bx))$}\big]-
\Ppr_{\SS\leftarrow \Dno ;\hspace{0.05cm}\bx} \big[\text{$A $ accepts $(\bx,\SS(\bx))$}\big]\le o(1),
\end{equation}
where $\bx$ denotes a sequence of $q$ points drawn from $\normal^n$ independently and
  $(\bx,\SS(\bx))$ denotes the $q$ labeled samples from $\SS$.
Theorem \ref{thm:2slb} follows directly from (\ref{hehe200}).

To draw a set $\SS\leftarrow \Dyes$, we sample a sequence of $N=2^{\sqrt{n}}$ points
  $\by_1,\ldots,\by_N$ from the sphere $S^{n-1}(r)$ of radius $r$ for some $r=\Theta(n^{1/4})$.
Each $\by_i$ defines a halfspace $\bh_i=\{x:x\cdot \by_i\le r^2\}$.
$\SS$ is then the intersection of all $\bh_i$'s. \red{(This is essentially a construction used by Nazarov \cite{Nazarov:03} to exhibit a convex set that has large Gaussian surface area, and used by \cite{KOS:07} to lower bound the sample complexity of learning convex sets under the Gaussian distribution.)} The most challenging part of the two-sided lower bound proof is to show that,
  with $q$ points $\bx_1,\ldots,\bx_q\leftarrow \normal^n$,
  the $q$ bits $\SS(\bx_1),\ldots,\SS(\bx_q)$ with $\SS\leftarrow \Dyes$ are ``almost'' independent.
More formally, the $q$ bits $\SS(\bx_1),\ldots,\SS(\bx_q)$ with $\SS\leftarrow \Dyes$
  have $o(1)$-total variation distance from
  $q$ independent bits with the $i$th bit drawn from the marginal distribution
  of $\SS(\bx_i)$ as $\SS\leftarrow \Dyes$.
On the other hand, it is relatively easy to define a distribution~$\Dno$ that satisfies~(ii)
  and at the same time, $\SS(\bx_1),\ldots,\SS(\bx_q)$ when $\SS\leftarrow\Dno$ has $o(1)$-total variation
  distance from the same product distribution.
(\ref{hehe200}) follows by combining the two parts.
}\vspace{-0.2cm}

\paragraph{Structural result.}  Our algorithms rely on a new structural result which we establish for convex sets in $\R^n$.  Roughly speaking, this result gives an upper bound on the Gaussian volume of the ``thickened surface'' of any bounded convex subset of $\R^n$; it is inspired by, and builds on, the classic result of Ball \cite{Ball:93} that upperbounds the Gaussian surface area of any convex subset of $\R^n$.\vspace{-0.2cm}

\paragraph{One-sided upper bound.} Our one-sided testing algorithm employs a ``gridding-based'' approach to decompose the relevant portion of $\R^n$ (namely, those points which are not too far from the origin) into a collection of disjoint cubes.  It draws samples and identifies a subset of these cubes as a proxy for the ``thickened surface'' of the target set; by the structural result sketched above, if the Gaussian volume of this thickened surface is too high, then the one-sided algorithm can safely reject (as the target set cannot be convex).  Otherwise the algorithm does random sampling to probe for points which are inside the convex hull of positive examples it has received but are labeled negative (there should be no such points if the target set is indeed convex, so if such a point is identified, the one-sided algorithm can safely reject).  If no such points are identified, then the algorithm accepts.\vspace{-0.2cm}

\paragraph{Two-sided upper bound.}  Finally, the main tool we use to obtain our two-sided testing algorithm is a \emph{learning} algorithm for convex sets with respect to the normal distribution over $\R^n.$  The main result of \cite{KOS:07} is an (improper) algorithm which learns the class of all convex subsets of $\R^n$ to accuracy $\eps$ using $n^{O(\sqrt{n}/\eps^2)}$ independent samples from $\normal^n$.  Using the structural result mentioned above, we show that this can be converted into a \emph{proper} algorithm for learning convex sets under $\normal^n$, with essentially no increase in the sample complexity.  Given this proper learning algorithm, a two-sided algorithm for testing convexity follows from the well-known result of \cite{GGR98} which shows that proper learning for a class of functions implies (two-sided) testability.

%% file: prelims.tex

\section{Preliminaries and Notation} \label{sec:prelims}

\paragraph{Notation.}
We use boldfaced letters such as $\bx, \boldf,\bA$, etc. to denote random variables (which may be real-valued, vector-valued, function-valued, set-valued, etc; the intended type will be clear from the context).
We write ``$\bx \leftarrow \calD$'' to indicate that the random variable $\bx$ is distributed according to probability distribution $\calD.$ Given $a,b,c\in \R$ we use $a=b\pm c$ to indicate that $b-c\le a\le b+c$.\vspace{-0.2cm}

\paragraph{Geometry.}
For $r >0$, we write $S^{n-1}(r)$ to denote the origin-centered sphere of radius $r$ in $\R^n$
and $\Ball(r)$ to denote the origin-centered ball of radius $r$ in $\R^n$, i.e.,
$$
S^{n-1}(r) = \big\{x \in \R^n: \|x\|=r\big\}\quad\text{and}\quad
\Ball(r) = \big\{x \in \R^n : \|x\| \leq r\big\},
$$
where $\|x\|$   denotes the $\ell_2$-norm $\|\cdot \|_2$ of $x\in \R^n$.
We also write $S^{n-1}$ for the unit sphere $S^{n-1}(1)$.

Recall that a set $C \subseteq \R^n$ is convex if $x,y \in C$ implies $\alpha\hspace{0.03cm}x + (1-\alpha) y \in C$ for all $\alpha\in [0,1].$
We write $\calC_\convex$ to denote the class of all convex sets in $\R^n.$
Recall that convex sets are Lebesgue measurable.
Given a set $C \subseteq \R^n$ we write $\Conv(C)$ to denote the convex hull of $C$.

For sets $A,B \subseteq \R^n$, we write $A + B$ to denote the Minkowski sum $\{a + b: a \in A\ \text{and}\ b \in B\}.$ For a set $A \subseteq \R^n$ and $r > 0$ we write $rA$ to denote the set $\{ra : a \in A\}$.
Given a point $a$ and a set $B\subseteq \R^n$, we use $a+B$ and $B-a$ to denote
  $\{a\}+B$ and $B+\{-a\}$ for convenience.  For a convex set $C$,
  we write $\partial C$ to denote its \emph{boundary}, i.e. the set of points $x \in \R^n$ such that for all $\delta > 0$,
  the set $x + \Ball(\delta)$ contains at least one point in $C$ and at least one point outside $C$.
\vspace{-0.2cm}

\paragraph{Probability.}
We use $\normal^n$ to denote the standard $n$-dimensional Gaussian distribution with zero mean and identity covariance matrix. We also recall that the probability density function for the one-dimensional Gaussian distribution is $$\varphi(x) = {\frac 1 {\sqrt{2 \pi}}}\cdot \exp(-x^2/2).$$
Sometimes we denote $\normal^n$ by $\gaus$ for convenience.
The squared norm $\|\bx\|^2$ of $\bx \leftarrow \normal^n$ is distributed according to the chi-squared distribution $\chi_n^2$
  with $n$ degrees of freedom.  The following tail bound for $\chi_n^2$ (see \cite{Johnstone01}) will be useful:

\begin{lemma} [Tail bound for the chi-squared distribution] \label{lem:johnstone}\label{johnstone}
Let $\bX \leftarrow \chi_n^2$.
Then we have
$$\Pr\big[|\bX-n| \geq tn\big] \leq e^{-(3/16)nt^2},\quad\text{for all $t \in [0, 1/2)$.}$$
\end{lemma}

All target sets $S\subseteq\R^n$ to be tested for convexity are assumed to be
  Lebesgue measurable  and we write $\Vol(S)$ to denote $\Pr_{\bx \leftarrow \gaus}[\bx \in S]$, the \emph{Gaussian volume} of $S\subseteq \R^n$.  Given  two Lebesgue measurable subsets $S,C \subseteq \R^n$, we view $\Vol(S \bigtriangleup C)$ as the \emph{distance} between $S$ and $C$, where $S\bigtriangleup C$ is the symmetric
    difference of $S$ and $C$.
Given $S\subseteq \R^n$, we abuse the notation and use $S$ to denote the
  indicator function of the set, so we may write  ``$S(x)=1$'' or ``$x \in S$'' to mean the same thing.

We say that a subset $\calC$ of $\calC_\convex$ is a \emph{$\tau$-cover of $\calC_\convex$} if for every $C \in \calC_\convex$, there exists a set $C' \in \calC$ such that $\Vol(C \bigtriangleup C') \leq \tau.$

\red{Given a convex set $C$ and a real number $h>0$, we let
  $C_h$ denote the set of points in $\R^n$ whose distance from $C$ do not exceed $h$.
We recall the following theorem of Ball \cite{Ball:93} (also see \cite{Nazarov:03}).

\begin{theorem}[\cite{Ball:93}]\label{balllemma}
For any convex set $C\subseteq \R^n$ and  $h>0$, we have
$$
\frac{\Vol(C_h\setminus C)}{h}\le 4\hspace{0.03cm}n^{1/4}.\vspace{-0.32cm}
$$
\end{theorem}
}

\paragraph{Sample-based property testing.}
Given a point $x\in \R^n$, we refer to $(x,S(x))\in \R^n\times \{0,1\}$ as a
  \emph{labeled sample} from a set $S\subseteq \R^n$.
A \emph{sample-based testing algorithm for convexity} is a randomized algorithm which is given
  as input an accuracy parameter $\eps>0$ and
  access to an oracle that, each time it is invoked, generates a labeled sample $(\bx,\targetset(\bx))$ from the unknown (Lebesgue measurable) \emph{target set}
  $\targetset \subseteq \R^n$ with $\bx$ drawn independently each time from $\gau$.
When run with any Lebesgue measurable $\targetset \subseteq \R^n,$ such an algorithm must output ``accept'' with probability at least 2/3 (over the draws it gets from the oracle and its own internal randomness) if $\targetset \in \calC_\convex$ and must output ``reject'' with probability at least $2/3$ if $\targetset$ is $\eps$-\emph{far} from being convex, meaning that for every $C \in \calC_\convex$ it is the case that $\Vol(\targetset \bigtriangleup C) \geq \eps$. (We also refer to an algorithm as an $\eps$-tester for convexity
  if it works for a specific accuracy parameter $\eps$.) Such a testing algorithm is said to be \emph{one-sided} if whenever it is run on a convex set $\targetset$ it always outputs ``accept;'' equivalently, such an algorithm can only output ``reject'' if the labeled samples it receives are not consistent with any convex set.  A testing algorithm which is not one-sided is said to be \emph{two-sided}.

Throughout the rest of the paper we reserve the symbol $\targetset$ to denote the unknown target set (a measurable subset of $\R^n$) that is being tested for convexity.  If $\targetset(x) = 1$ then we say that $x$ is a \emph{positive point}, and if $\targetset(x)=0$ we say $x$ is a \emph{negative point}.

Given a finite set $T$ of labeled samples $(x,b)$ with $x\in \R^n$ and  $b \in \{0,1\}$,
  we say $x$ is a \emph{positive} point in $T$ if $(x,1)\in T$ and is a \emph{negative} point
  in $T$ if $(x,0)\in T$.
We use
  $T^+$ to denote the set of positive points $\{x: (x,1) \in T\}$, and $T^-$ to denote the
  set of negative points $\{x: (x,0) \in T\}.$

%% file: structural.tex

\section{A useful structural result:  Bounding the volume of the\\ thickened boundary of bounded convex bodies} \label{sec:structural}

For a \emph{bounded} convex set $C$ in $\R^n$ (i.e., $\sup_{c\in C}\|c\|\le K$ for some real $K$) we may view $\partial C + \Ball(\alpha)$ as the ``$\alpha$-\emph{thickened boundary}'' of $C$.  In this section, we use Theorem~\ref{balllemma} of \cite{Ball:93} to give an upper bound on the volume of the $\alpha$-thickened boundary of such a set:

\begin{theorem}\label{thm:surfacevolume}
If $C\subset \R^n$ is convex and $\sup_{c\in C} \|c\| \leq K$ for some $K >1$, then  we have
$$\Vol\big(\del C + \Ball(\alpha)\big) \leq 20\hspace{0.03cm} n^{\red{5/8}}\hspace{0.03cm}K \sqrt{\alpha},\quad\text{\red{for any $0 < \alpha < n^{-3/4}$.}}$$
\end{theorem}

Having such a bound will be useful to us in two different contexts.  First, it plays an important role in the proof of correctness of our one-sided algorithm for testing convexity (see Section~\ref{sec:1sub}).  Second, as an easy consequence of the theorem, we get an algorithm which, for any $\tau>0$, constructs a $\tau$-cover of $\calC_\convex$ (this is Corollary~\ref{cor:epscover}, which we defer to later as its proof employs a ``gridding'' argument which we introduce in Section~\ref{sec:1sub}).  This cover construction algorithm plays an important role in our two-sided algorithm for testing convexity (see Section~\ref{sec:2sub}).

\subsection{Proof of Theorem~\ref{thm:surfacevolume}}

Let $C\subset \R^n$ be a bounded convex set that satisfies $\sup_{c\in C} \|c\| \leq K$ for some $K >1$.

The proof has two cases and uses Lemmas
  \ref{lem:noball}, \ref{lem:small}, and \ref{lem:large} to be proved later.

\medskip
\noindent {\bf Case I:  } $C$ contains no ball of radius $\rho := \sqrt{\alpha}/n^{3/8}$.  In this case we have
\begin{align*}
\Vol(\del C + \Ball(\alpha)) \leq \Vol(C + \Ball(\alpha))
& \leq 2(n\rho + \alpha) \tag*{(Lemma~\ref{lem:noball})}\\
& \leq 3\hspace{0.03cm}n^{5/8}\hspace{0.03cm}\sqrt{\alpha} \tag*{\red{(using $\alpha < n^{-3/4}$)}}\\
& < 20 \hspace{0.03cm}n^{5/8}\hspace{0.03cm}K\hspace{0.03cm}\sqrt{\alpha} \tag*{(using $K > 1$)}
\end{align*}

\noindent {\bf Case II:}  $C$ contains some ball of radius $\rho$. We let $z^*$ be the center of
  such a ball and let $D = \del C +$ $\Ball(\alpha).$
To upperbound $\Vol(D)$, we define a set that contains $D$ and then upperbound its volume.

To this end, we first shift $C$ to get $C'=C-z^*$ (so that the ball of radius $\rho$ is now centered
  at the origin).
By triangle inequality we have $\sup_{c\in C'} \|c\|\le 2K$.
Let $\beta = n^{3/8}\sqrt{\alpha} = \alpha/\rho$,
and observe that since $\alpha<n^{-3/4}$ we have $\beta < 1.$
Let $D'=D-z^*=\del C'+\Ball(\alpha)$. By Lemma~\ref{lem:small}, we have
$$C_0' := (1-\beta)C'=(1-\beta)(C-z^*)$$
 contains no point of $D'$,
 and then \red{by Lemma~\ref{lem:large} the set
 $C_1' := (C_0')_h$ with $h={4}\beta K + \alpha$ contains~all of $D'$.}\footnote{Recall that $(C_0')_h$
 is the set of all points that have distance at most $h$ to $C_0'$. Also note that
   the coefficient of $\beta K$ in our choice of $h$ is $4$ instead of $2$ since we have $\sup_{c\in C'}\|c\|\le 2K$ instead of $K$.}
As a result, $D'\subseteq C_1'\setminus C_0'$ and it suffices to upperbound
  $\Vol(z^*+C_1'\setminus C_0')$, which is at most $4hn^{1/4}$ by Theorem \ref{balllemma}
  (since $C_0'$ is convex).
Combining everything together, we have
\begin{align*}
\Vol(D)  \leq \Vol(z^*+ C_1' \setminus C_0') \le
(4 \beta K + \alpha)\hspace{0.03cm} (4n^{1/4} ) \leq
  20\hspace{0.03cm}n^{5/8}\hspace{0.03cm}K\hspace{0.03cm}\sqrt{\alpha}.\end{align*}
(again using $K>1$ and $\alpha < n^{-3/4}$ for the last inequality). \qed

\medskip

It remains to prove Lemmas~\ref{lem:noball}, \ref{lem:small}, and \ref{lem:large}.  We prove these lemmas in Appendix~\ref{ap:lemmas}.

%% file: 1sub.tex

\section{One-sided upper bound:  Proof of Theorem~\ref{thm:1sub}} \label{sec:1sub}

Recall Theorem~\ref{thm:1sub}:

\begin{reptheorem}{thm:1sub}
For any $\eps > 0$, there is a one-sided sample-based $\eps$-tester for convexity over
  $\normal^n$ which uses $(n/\eps)^{O(n)}$ samples.
\end{reptheorem}

In Section~\ref{sec:setup} we show that it suffices to test convex bodies contained in a large ball $B$ centered at the origin (rather than all of $\R^n$) and give some useful preliminaries.  \red{Section~\ref{sec:bound-vol} then builds~on Theorem~\ref{thm:surfacevolume} (the upper bound on the volume of the ``thickened boundary'' of any \blue{bounded} convex body) to give an upper bound, in the case that $\targetset$ is convex \blue{and contained in $B$}, on the total volume of certain ``boundary cubes'' (defined in Section~\ref{sec:setup}).} In Section~\ref{sec:alg} we present the one-sided testing algorithm and establish its correctness, thus proving Theorem~\ref{thm:1sub}.

\subsection{Setup} \label{sec:setup}

Let $n'$ be the following parameter (that depends on both $n$ and $\eps$):
$$
n':= \left(n + 4\sqrt{n\ln(4/\eps)}\right)^{1/2}.
$$
Let $\calC'_\convex$ denote the set of convex bodies in $\R^n$ that are contained in $\Ball(n')$, equivalently,
\[
\calC'_\convex = \big\{ C \cap \Ball(n'): C \in \calC_\convex\big\}.
\]

We prove the following claim that helps us focus on testing of $\calC'_\convex$ instead $\calC_\convex$.

\begin{claim} \label{claim:reduction}
Suppose that there is a one-sided sample-based $\eps$-testing algorithm $A'$ which, given any Lebesgue measurable target set $\targetset$ contained in $\Ball(n')$, uses $(n/\eps)^{O(n)}$ samples drawn from $\normal^n$ to test whether $\targetset \in \calC'_\convex$ versus $\targetset$ is $\eps$-far from $\calC'_\convex$.  Then this implies Theorem~\ref{thm:1sub}.
\end{claim}

\begin{proof}
Given $A'$ for $\calC'_\convex$, we consider an algorithm $A$ which works as follows to test whether~an arbitrary Lebesgue measurable subset $\targetset$ of $\R^n$ is convex or $\eps$-far from $\calC_\convex$: algorithm $A$ runs $A'$ with parameter $\eps/2$, but with the following modification: each time $A'$ receives from the oracle a labeled sample $(x,b)$ with $x \notin \Ball(n')$, it replaces the label $b$ with $0$ and gives the modified labeled sample to $A'$.   When the run of $A'$ is complete $A$ returns the output of $A'$.

If $\targetset \subseteq \R^n$ is the target set, then it is clear that the above modification results in running $A'$ on $\targetset \cap \Ball(n')$.  If $\targetset$ is convex, then $\targetset \cap \Ball(n')$ is also convex. As $A'$ commits only one-sided error, it will always output ``accept,'' and hence so will $A$. On the other hand, suppose that  $\targetset$ is $\eps$-far from $\calC_\convex$. We claim that
$\Vol(\Ball(n')) \geq 1 - \eps/4$ (this will be shown below); given this claim, it must be the case that $\targetset \cap \Ball(n')$ is at least $(3\eps/4)$-far from $\calC_\convex$ and at least $(3\eps/4)$-far from $\calC'_\convex$ as well.  Consequently $A'$ will output ``reject'' with probability at least $2/3$, and hence so will $A$.

To bound
$\Vol(\Ball(n'))$, observe that it is the probability that an $\bx \leftarrow\normal^n$ has
$$\|\bx\|^2\leq n + 4\sqrt{n\ln(4/\eps)}.$$
It follows from Lemma~\ref{lem:johnstone} that the
probability is at least $1-\eps/4$ as claimed.
\end{proof}

Given Claim~\ref{claim:reduction}, it suffices to prove the following slight variant of Theorem~\ref{thm:1sub}:

\begin{theorem} \label{thm:1submodif}
There is a one-sided sample-based $\eps$-testing algorithm $A'$ which, given any Lebesgue measurable target set $\targetset$ contained in $\Ball(n')$, uses $(n/\eps)^{O(n)}$ samples from $\normal^n$ to test whether $\targetset \in \calC'_\convex$ versus $\targetset$ is $\eps$-far from $\calC'_\convex$.
\end{theorem}

In the rest of this section we prove Theorem~\ref{thm:1submodif}. We start with some terminology and concepts that we use in the description and analysis of our algorithm.  \red{Some of the notions that we introduce below, such as the notions of ``boundary'' cubes and ``internal'' cubes, are inspired by related notions that arise in earlier works such as \cite{Kern,Raskhodnikova:03}.}

Fix $\ell := \eps^3/n^4$ in the rest of the section, and let $\Cube_0$ denote the following set
\[
\Cube_0 :=  [-\ell/2,\ell/2)^{n} \subset \R^n
\]
of side length $\ell$ that is centered at the origin.  We say that a \emph{cube} is a subset of $\R^n$ of the form $\Cube_0 + \ell \cdot (i_1,\dots,i_n)$, where each $i_j \in \Z$, which contains at least one point of $\Ball(2n').$
We use $\CubeSet$ to denote the set of all such cubes.

  It is easy to see that
\[
\Ball(n') \subset \text{union of all cubes in $\CubeSet$} \subset \Ball(2n' +  \ell \sqrt{n}) \subset \Ball(3n').
\]

Fix an $\targetset \subseteq \Ball(n')$ as the target set being tested for membership in $\calC'_\convex.$  Additionally fix a finite set $T=\{(x^1,S(x^1)),\dots,(x^M,S(x^M))\}$ of labeled samples according to $\targetset$, for some~positive integer $M$.  (The set $T$ will correspond to the set of labeled samples that the testing algorithm receives.)  We classify cubes in the $\CubeSet$ based on $T$ in the following way:

\begin{figure}
\centering
\includegraphics[width=7cm]{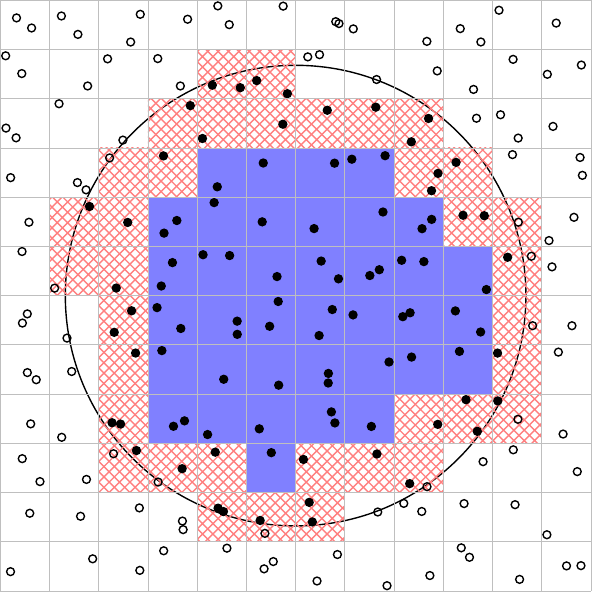}\vspace{0.15cm}
\caption{A 2D example of the different types of cubes induced by a set of labeled samples.~The target set $S$ is a disk, and the solid and hollow dots are positive and negative samples, respectively. The hollow, hatched, and shaded boxes are external, boundary, and internal cubes, respectively. }
\label{fig-boxes}
\end{figure}

\begin{flushleft}
\begin{itemize}
\item A cube \(\Cube\) is said to be an \emph{external cube} if $\Cube \cap T^+ = \emptyset$ (i.e., no positive point of $T$ lies in $\Cube$). We let $EC$ denote the union of all the external cubes.

\item Any cube which is not an external cube (equivalently, any cube that contains at least one positive point of $T$) is said to be a \emph{positive cube}.

\item We say that two cubes $\Cube,\Cube'$ are \emph{adjacent} if for any $\kappa>0$ there exist $x \in \Cube$ and $y \in \Cube'$ that have Euclidean distance at most $\kappa$ (in other words, two cubes are adjacent if their closure ``touch anywhere, even only at a vertex;'' note that each cube is adjacent to itself).  If a cube is both (i) a positive cube and (ii) is adjacent to a cube (including itself) that contains at least one negative point of $T$, then we call it a \emph{boundary cube}.  We use $BC$ to denote the union of all boundary cubes.

\item We say that a positive cube which is not a boundary cube is an \emph{internal cube}.  (Equivalently, a cube is internal if and only if \red{it contains at least one positive point} and all the points in $T$ that are contained in any of its adjacent cubes, including itself, are positive.) We use $IC$ to denote the union of all internal cubes.
\end{itemize}
\end{flushleft}
We note that since each  cube is either external, internal, or boundary, the set $\Ball(n')$ is contained in the (disjoint) union of $EC,BC$ and $IC.$ Figure~\ref{fig-boxes} illustrates the different types of cubes.

We will use the following useful property of internal cubes:

\begin{lemma}\label{lem:i}  Suppose a finite set of labeled samples $T$ is such that every cube in $\CubeSet$ contains at least one point of $T$.  Then every internal cube is contained in $\Conv(T^+).$
\end{lemma}

The lemma is a direct consequence of the following claim by setting $H=T^+$:
\begin{claim} \label{claim:internal}
Let $H \subseteq \R^n$ be any set that contains at least one point in each cube that is adjacent to $\Cube_0$. Then $\Cube_0$ is contained in $\Conv(H)$.
\end{claim}

\begin{proof}
We prove the claim by induction on the dimension $n$. When $n=1$ the claim is trivial since $\Cube_0$ is simply the interval $[-\ell/2,\ell/2)$ and by assumption, there is at least one point of $H$ in $[-3\ell/2,-\ell/2)$ and at least one point of $H$ in $[\ell/2,3\ell/2).$

For $n > 1$, let $P = \{p \in H \mid p_n \geq {\ell}/{2}\}$ and $P' = \{p' \in H \mid p_n' \leq  {-\ell}/{2}\}$ be two subsets of $H$. Intuitively, the convex hulls of $P$ and $P'$ ``cover'' $\Cube_0$ on both sides (by induction), so the convex hull of their union will contain the whole $\Cube_0$.  More formally, let $x$ be any point in $\Cube_0$. By projecting $P,P'$ and $x$ onto the first $n-1$ dimensions and using the inductive hypothesis\footnote{Observe that after projecting out the last coordinate, the assumed property of $H$ (that it  has at least one sample point in each adjacent cube) will still hold in $n-1$ dimensions.}, we can find points $y \in \Conv(P)$ and $y' \in \Conv(P')$ such that $y_i = y_i' = x_i$ for all $i\in [n-1]$. Since we~have $p_n \geq 1/{2}$ and $p_n' \leq -1/{2}$ for all $p \in P$ and $p' \in P'$, respectively, it follows directly that $y_n \geq {1}/{2}$ and $y_n' \leq -{1}/{2}$. As $x\in \Cube_0$, $x$ is on the line segment between $y$ and $y'$ and thus is in the convex hull of $H$. Hence all of $\Cube_0$ is contained in $\Conv(H)$.
\end{proof}

\subsection{Bounding the total volume of boundary cubes} \label{sec:bound-vol}

Before presenting our algorithm we record the following useful corollary of Theorem~\ref{thm:surfacevolume}, which allows the one-sided tester to reject bodies as non-convex if it detects too much volume in boundary cubes.
\red{(Note that we do not assume below that $T$ satisfies the condition of Lemma \ref{lem:i}, i.e.,
  that $T$ has at least one point in each cube in $\CubeSet$, though this will be
  the case when we use it later.)}

\begin{corollary}\label{cor:smallconvbdry}

Let $\targetset$ be a convex set in $\calC'_\convex$ and
  $T$ be any finite set of labeled samples according to $\targetset$,
  which defines sets $EC,IC$ and $BC$ as discussed earlier.  Then we have
  $$\Vol(BC) \leq 20\hspace{0.03cm}n^{5/8}\hspace{0.03cm}n'\hspace{0.03cm}\sqrt{2\ell\sqrt{n}} =\red{o(\eps)}.$$

\end{corollary}
\begin{proof}

Let $\Cube$ be a boundary cube. Then by definition, there is a positive point of $T$ (call~it~$t$) in $\Cube$, and there is a $\Cube'$ adjacent to $\Cube$ that contains a negative point of $T$ (call it $t'$).~It follows that there must be a boundary point of $\del S$ (call it $t^*$) in the segment between $t$ and $t'$, and we have $\Cube \in t^* + \Ball(2 \ell \sqrt{n}).$  It follows that $BC \subseteq \del S + \Ball(2 \ell \sqrt{n})$, and hence
$$\Vol(BC) \leq \Vol\big( \del S + \Ball(2 \ell \sqrt{n})\big)\le 20 \hspace{0.03cm}n^{5/8}\hspace{0.03cm} n'  \sqrt{2\ell\sqrt{n}} =\red{o(\eps)}$$ by Theorem~\ref{thm:surfacevolume} (and using $\ell\sqrt{n}\ll n^{-3/4}$ by our choice of $\ell=\eps^3/n^4$).

\end{proof}

\subsection{The one-sided testing algorithm} \label{sec:alg}

Now we describe and analyze
  the one-sided testing algorithm $A'$ mentioned in Theorem~\ref{thm:1submodif}.
Algorithm $A'$ works by performing $O(1/\eps)$ independent runs of the algorithm $A^*$,
  which we describe in Figure \ref{fig:main}.  If any of the runs of $A^*$ output ``reject'' then algorithm $A'$ outputs ``reject,'' and otherwise it outputs ``accept.''

\begin{figure}[t!]
\begin{framed}\vspace{0.1cm}
\noindent {\bf Algorithm $A^*$:} Given access to independent draws $(\bx,S(\bx))$ where $\bx \leftarrow \normal^n$ and the\\ target set $S$ is a Lebesgue measurable set that \red{is contained in $\Ball(n')$}.
\begin{flushleft}\begin{enumerate}

\item Draw a set $\bT$ of $s:= (n/\eps)^{O(n)}$ labeled samples $(\bx,\targetset(\bx))$, where each $\bx \leftarrow \normal^n$.\vspace{-0.05cm}

 \item If any cube does not contain a point of $\bT$, then halt and output ``accept.''\vspace{-0.05cm}

 \item If $\Vol(BC)\ge\eps/4$ (the volume of the union of boundary cubes),
  halt and output ``reject.''\vspace{-0.05cm}

 \item Define $\bI \subseteq \R^n$ to be $\Conv(\bT^+)$, the convex hull of all positive points in $\bT$.\vspace{-0.05cm}

 \item Draw a single fresh labeled sample $(\by,S(\by))$, where $\by\leftarrow \normal^n$.   If $\by \in \bI$ but $S(\by)=0$ then halt and output ``reject.'' Otherwise, halt and output ``accept.''\vspace{-0.16cm}
\end{enumerate}\end{flushleft}
\end{framed}\vspace{-0.25cm}\caption{Description of the algorithm $A^*$}\label{fig:main}\end{figure}

In words, Algorithm~$A^*$ works as follows:  first, in Step~1 it draws enough samples so that (with very high probability) it will receive at least one sample in each cube (if the low-probability event that this does not occur takes place, then the algorithm outputs ``accept'' since it can only reject if it is impossible for $\targetset$ to be convex).  If the region ``close to the boundary'' of $\targetset$ (as measured by $\Vol(BC)$ in Step~3) is too large, then the set cannot be convex (by Corollary \ref{cor:smallconvbdry}) and the algorithm rejects.  Finally, the algorithm checks a freshly drawn point; if this point is in the convex hull of the positive samples but is labeled negative, then the set cannot be convex and the algorithm rejects. Otherwise, the algorithm accepts.

To establish correctness and prove Theorem~\ref{thm:1submodif} we must show that (i) algorithm $A^*$ never rejects if the target set $\targetset$ is a Lebesgue measurable set that belongs to $\calC'_\convex$, and (ii) if $\targetset$ is $\eps$-far from $\calC'_\convex$ then algorithm $A^*$ rejects with probability at least $\Omega(\eps).$
Part (i) is trivial as $A^*$ only rejects if either
(a) $\Vol(BC) \ge \eps/4$ or (b) step~5 identifies a negative point in the convex hull of the positive points in $\bT$.
For both cases we conclude (using Corollary \ref{cor:smallconvbdry} for (a))
  that $S\notin \calC'_\convex$.

For (ii) suppose that $\targetset$ is $\eps$-far from $\calC'_\convex$.  Let
$E$ be the following event (over the draw of $\bT$):
\begin{flushleft}\begin{quote}
Event $E$: Every cube in $\CubeSet$ contains at least one point of $\bT$ (so the\\ algorithm does not
  accept in Step 2) and moreover, every $\Cube$ with
$$\frac{\Vol(\Cube \cap \targetset)}{\Vol(\Cube)} \geq \e/4$$
contains at least one positive point in $\bT$ and thus, is not external.
\end{quote}\end{flushleft}
It is easy to show that the probability mass of each cube in $\CubeSet$ is
  at least $( \eps/n)^{O(n)}$ (since its volume is $(\eps/n)^{O(n)}$ and the density
  function of the Gaussian is at least $(1/\eps)^{O(n)}$ using our~choice of $n'$),
  it follows from a union bound over $\CubeSet$ that,
  for a suitable choice of $s=(n/\eps)^{O(n)}$ (with a
  large enough coefficient in the exponent),
  $E$ occurs with probability $1-o(1)$.
Assuming that $E$ occurs, we show below that either $\Vol(BC)\ge \eps/4$ or $A^*$ rejects in Step 5
  with probability $\Omega(\eps)$.

For this purpose, we assume below that both $E$ occurs and $\Vol(BC)<\eps/4$.
Note that the set $I$ is convex and is contained in $\Ball(n')$.
Thus it belongs to $\calC'_\convex$ and consequently $\Vol(I \bigtriangleup \targetset) \geq \eps$
  (since $\targetset$ is assumed to be $\eps$-far from $\calC'_\convex$), which implies that
$$
\Vol(S\setminus I)+\Vol(I\setminus S)\ge \eps.
$$
It suffices to show that $\Vol(\targetset \setminus I) \leq \e/2$, since $\Vol(I \setminus \targetset)$
  is exactly the probability that algorithm $A^*$ rejects in Step 5.
To see that $\Vol(\targetset \setminus I) \leq \e/2$, observe that by Lemma~\ref{lem:i}, $\Vol(\targetset \setminus I)$ is at most $\Vol(\targetset \cap BC) + \Vol(\targetset \cap EC)$.
On the one hand, $\Vol(S\cap BC)\le \Vol(BC)<\eps/4$ by assumption.
On the other hand, given the event $E$, every external cube has at most
  $(\eps/4)$-fraction of its volume in $S$ and thus, $\Vol(S\cap EC)\le \eps/4$
  (as the total volume of $EC$ is at most $1$).
Hence $\Vol(S\setminus I)\le \eps/2$.

This concludes the proof of Theorem~\ref{thm:1submodif}.

%% file: 2slb.tex

\section{Two-sided lower bound} \label{sec:2slb}

We recall Theorem \ref{thm:2slb}:

\begin{reptheorem}{thm:2slb}
\red{There exists a positive constant $\eps_0$ such that any
  two-sided sample-based algorithm that is an $\eps$-tester for convexity over $\normal^n$
  for some $\eps\le \eps_0$ must use $2^{\Omega(\sqrt{n})}$ samples.}
\end{reptheorem}

\newcommand{\D}{\mathcal{D}}
\newcommand{\capp}{\mathrm{cap}}
\newcommand{\fsa}{\mathrm{fsa}}

Let $q=2^{0.01\sqrt{n}}$ and let $\eps_0$ be a~positive constant to be specified later.
To prove Theorem \ref{thm:2slb}, we show  that  no sample-based, $q$-query (randomized) algorithm $A$ can achieve the following goal:
\begin{flushleft}\begin{quote}
Let $S\subset \R^n$ be a target set that is Lebesgue measurable.
Let $\bx_1,\ldots,\bx_q$ be a sequence of $q$ samples drawn from $\gau$.
Upon receiving $((\bx_i,S(\bx_i)):i\in [q])$,
    $A$ accepts with probability at least $2/3$ when $S$ is convex
  and rejects with probability at least $2/3$ when $S$ is $\eps_0$-far from convex.

\end{quote}\end{flushleft}
Recall that a pair $(x,b)$ with $x\in \R^n$ and $b\in \{0,1\}$ is a {labeled sample}.
Thus, a sample-based algorithm $A$ is simply a randomized map from a sequence of $q$ labeled samples to $\{\text{``accept'',``reject''}\}$.

\subsection{Proof Plan}

Assume for contradiction that there is a
   $q$-query (randomized) algorithm $A$ that accomplishes the task above.
In Section \ref{sec:dist} we define two probability distributions $\Dyes$ and $\Dno$ such that
  (1) $\Dyes$ is a distribution over convex sets in $\R^n$ ($\Dyes$
  is a distribution over certain convex polytopes that are the intersection of many randomly drawn halfspaces), and (2)
  $\Dno$ is a probability distribution over  sets in $\R^n$ that are Lebesgue measurable
  ($\Dno$ is actually supported over a finite number of measurable sets in $\R^n$) such that
  $\SS\leftarrow \Dno$ is $\eps_0$-far from convex   with probability at least $1-o(1)$.

Given a sequence $x=(x_1,\ldots,x_q)$ of points, we abuse the notation and write
  $$S(x)=(S(x_1),\ldots,S(x_q))$$
  and use $(x,S(x))$ to denote the sequence of $q$ labeled samples
  $(x_1,S(x_1)),\ldots,(x_q,S(x_q))$.
It then follows from our assumption on $A$ that
\begin{align*}
\Ppr_{\SS\leftarrow \Dyes;\hspace{0.05cm}\bx \leftarrow (\gaus)^q} \big[\text{$A$ accepts $(\bx,\SS(\bx))$}\big] &\ge 2/3\quad  \text{and} \\[0.5ex]
\Ppr_{\SS\leftarrow \Dno ;\hspace{0.05cm}\bx \leftarrow (\gaus)^q} \big[\text{$A$ accepts $(\bx,\SS(\bx))$}\big] &\le 1/3+o(1).
\end{align*}
where we use $\bx\leftarrow (\gaus)^q$ to denote a sequence of $q$ points sampled independently from $\gaus$
  and we usually skip the $\leftarrow (\gaus)^q$ part in the subscript when it is clear from the context.
Since $A$ is~a mixture of deterministic algorithms,
  there exists a deterministic sample-based, $q$-query algorithm $A'$ (equivalently,
  a deterministic map from sequences of $q$ labeled samples to $\{\textsf{``Yes''}, \textsf{``No''}\}$) with
\begin{equation}\label{hehe2}
\Ppr_{\SS\leftarrow \Dyes;\hspace{0.05cm}\bx} \big[\text{$A'$ accepts $(\bx,\SS(\bx))$}\big]-
\Ppr_{\SS\leftarrow \Dno ;\hspace{0.05cm}\bx} \big[\text{$A'$ accepts $(\bx,\SS(\bx))$}\big]\ge 1/3-o(1).
\end{equation}
Let $\Eyes$ (or $\Eno$) be the distribution of $(\bx,\SS(\bx))$, where $\bx\leftarrow (\gaus)^q$ and
  $\SS\leftarrow \Dyes$ (or $\SS\leftarrow \Dno$, respectively). Both of them are distributions over sequences of $q$
  labeled samples.
Then the LHS of (\ref{hehe2}), for any deterministic sample-based, $q$-query algorithm $A'$, is
  at most the total variation distance between $\Eyes$ and $\Eno$.
We prove the following key lemma, which leads to a contradiction.

\begin{lemma}\label{maintechlemma}
The total variation distance between $\Eyes$ and $\Eno$ is $o(1)$.
\end{lemma}

To prove Lemma \ref{maintechlemma}, it is convenient for us to introduce a third distribution $\S$ over
  sequences of $q$ labeled samples,
  where $(\bx,\bb)\leftarrow \S$ is drawn by first sampling a sequence of $q$ points $\bx=(\bx_1,\ldots,\bx_q)$
  from $\gaus$ independently  and then for each $\bx_i$, its label $\bb_i$ is set to be $1$ independently with
  a probability that depends only on $\|\bx_i\|$
(see Section \ref{sec:dist}).
Lemma \ref{maintechlemma} follows from the following two lemmas by the  triangle inequality.
\begin{lemma}\label{techlemma1}
The total variation distance between $\Eno$ and $\Eno^*$ is $o(1)$.
\end{lemma}

\begin{lemma}\label{techlemma2}
The total variation distance between $\Eyes$ and $\Eno^*$ is $o(1)$.
\end{lemma}

The rest of the section is organized as follows.
We define the distributions $\Dyes,\Dno$ (which are used to define $\Eyes$ and $\Eno$) as well as $\S$ in Section \ref{sec:dist}
  and prove the necessary properties about $\Dyes$ and $\Dno$ as well as Lemma \ref{techlemma1}.
We prove Lemma \ref{techlemma2} in Sections \ref{finalsec} and \ref{hehefinal}.

\subsection{The Distributions}\label{sec:dist}

Let $r=\Theta(n^{1/4})$ be a parameter to be fixed later, and let $N=2^{\sqrt{n}}$.
We start with the definition of $\Dyes$.
A random set $\SS\subset  \R^n$ is drawn from $\Dyes$ using the following procedure:
\begin{flushleft}\begin{enumerate}
\item We sample a sequence of $N$ points $\by_1,\ldots,\by_N$ from $\red{S^{n-1}(r)}$ independently and uniformly at random.
Each point $\by_i$ defines a halfspace $$\bh_i=\big\{x\in \R^n: x\cdot \by_i\le r^2\big\}.$$
\item The set $\SS$ is then the intersection of $\bh_i$, $i\in [N]$ (this is always nonempty as
indeed $\Ball(r)$ is contained in $\SS$).
\end{enumerate}\end{flushleft}
It is clear from the definition that $\SS\leftarrow \Dyes$ is always a convex set.

Next we define $\red{\S}$ (instead of $\red{\Dno}$), a distribution over \red{sequences of $q$} labeled samples $(\bx,\bb)$.
To this end, we use $\Dyes$ to define a function \(\rho: \R_{\ge 0} \to [0, 1]\) as follows:
$$
\rho(t) = \Ppr_{\SS \leftarrow \D_\text{yes}}\Big[(t, 0, \ldots,0) \in \SS\Big].
$$
Due to the symmetry of $\Dyes$ and $\gaus$, the value $\rho(t)$ is indeed the probability that a point $x\in \R^n$ at distance $t$
  from the origin lies in $\SS\leftarrow \Dyes$.
To draw a sequence of $q$ labeled samples $(\bx,\bb)\leftarrow \S$, we first independently draw $q$ random points $\bx_1,\dots,\bx_q \leftarrow \gaus$ and then independently set each $\bb_i=1$ with probability $\rho(\|\bx_i\|)$ and $\bb_i=0$ with probability $1-\rho(\|\bx_i\|)$.

Given $\Dyes$ and $\red{\S}$,
  Lemma \ref{techlemma2} shows that information-theoretically no sample-based algorithm
  can distinguish a sequence of $q$ labeled samples $(\bx,\bb)$ with $\SS\leftarrow \Dyes$,
  $\bx\leftarrow (\gaus)^q$, and $\bb=\SS(\bx)$
  from a sequence of $q$ labeled samples drawn from $\red{\S}$.
While the marginal distribution of each labeled sample is the same for the two cases,
  the former is generated in a correlated fashion using the underlying random convex $\SS\leftarrow \Dyes$ while the latter
  is generated independently.

Finally we define the distribution $\red{\Dno}$, prove Lemma \ref{techlemma1}, and
  show that a set drawn from $\red{\Dno}$ is far from convex with high probability.
To define $\red{\Dno}$, we let $M\ge 2^{\sqrt{n}}$ be a large enough integer to be specified later.
With $M$ fixed, we use
$$
0=t_0<t_1<\cdots<t_{M-1}<t_M=\red{2\sqrt{n}}
$$
to denote a sequence of numbers such that  the origin-centered ball $\Ball({\red{2\sqrt{n}}})$ is partitioned into
  $M$ \emph{shells} $\Ball(t_i)\setminus \Ball(t_{i-1})$, $i\in [M]$,
  and all the $M$ shells have the same probability mass under $\gaus$.
By spherical coordinates, it means that
  the following integral takes the same value for all $i$:
\begin{equation}\label{tttt}
\int_{t_{i-1}}^{t_i} \phi(x,0,\ldots,0) x^{n-1} dx,
\end{equation}
where $\phi$ denotes the density function of $\gaus$.
We show below that when $M$ is large enough, we have
\begin{equation}\label{ofofof}
|\rho(x)-\rho(t_i)|\le 2^{-\sqrt{n}},
\end{equation}
for any $i\in [M]$ and any $x\in [t_{i-1},t_i]$.
We will fix such an $M$ and use it to define $\red{\Dno}$. (Our results are not affected by the size of $M$ as a function of $n$; we only need it to be finite, given $n$.)

To show that (\ref{ofofof}) holds when $M$ is large enough, we need the continuity of the function $\rho$, which follows
  directly from the explicit expression for $\rho$ given later in
  (\ref{eq:rhodef}).

\begin{lemma}\label{continuity}
The function $\rho:\R_{\ge 0}\rightarrow [0,1]$ is continuous.
\end{lemma}

Since $\rho$ is continuous, it is continuous over $[0,\red{2\sqrt{n}}]$.
Since $[0,\red{2\sqrt{n}}]$ is compact,
  $\rho$ is also uniformly continuous over $[0,\red{2\sqrt{n}}]$.
Also note that $\max_{i\in [M]} (t_{i}-t_{i-1})$ goes to $0$ as $M$ goes to $+\infty$.
It follows that (\ref{ofofof}) holds when $M$ is large enough.

With $M\ge 2^{\sqrt{n}}$ fixed,
  a random set $\SS\leftarrow \red{\Dno}$ is drawn as follows.
We start with $\SS=\emptyset$ and for each $i\in [M]$,
  we add the $i$th shell $\Ball(t_i)\setminus \Ball(t_{i-1})$ to $\SS$ independently
  with probability $\rho(t_i)$.
Thus an outcome of $\SS$ is a union of some of the shells and
  $\red{\Dno}$ is supported over $2^M$  different sets.

Recall the definition of $\Eyes$ and $\Eno$ using $\Dyes$ and $\Dno$. We now prove Lemma \ref{techlemma1}.

\begin{proof}[Proof of Lemma \ref{techlemma1}]
Let $x=(x_1,\ldots,x_q)$ be a sequence of $q$ points in $\R^n$.
We say $x$ is \emph{bad} if either (1) at least one point lies outside of $\Ball(2\sqrt{n})$ or (2)
  there are two points that lie in the same shell of $\Dno$; we say $x$ is \emph{good} otherwise.
We first claim that $\bx\leftarrow (\gaus)^q$ is bad with probability $o(1)$.
To see this, we have from Lemma \ref{johnstone} that event (1) occurs with probability $o(1)$, and
  from $M\ge 2^{\sqrt{n}}$ and $q=2^{0.01\sqrt{n}}$ that event (2) occurs with probability $o(1)$.
The claim follows from a union bound.

Given that $\bx\leftarrow (\gaus)^q$ is good with probability $1-o(1)$, it suffices to show that
  for any good $q$-tuple $x$, the total variation distance between
  (1)  $\SS(x)$ with $\SS\leftarrow \Dno$ and (2)
  $\bb = (\bb_1,\dots,\bb_q)$ with each bit $\bb_i$ being $1$ with probability $\rho(\|x_i\|)$ independently, is $o(1)$.
Let $\ell_i \in [M]$ be the index of the shell that $x_i$ lies in.
Since $x$ is good (and thus, all points lie in different shells),
  $\SS(x)$ has the $i$th bit being $1$ independently with probability $\rho(t_{\ell_i})$;
  for the other distribution, the probability is $\rho(\|x_i\|)$.
Using the subadditivity of total variation distance (i.e., the fact that
  the $\dtv$ between two sequences of independent random variables
  is upper bounded by the sum of the $\dtv$ between each pair) as well as (\ref{ofofof}), we have
$
\smash{\dtv (\SS(x),\bb )\le q\cdot 2^{-\sqrt{n}}=o(1).}
$
This finishes the proof.
\end{proof}

The next lemma shows that $\SS\leftarrow \Dno$ is $\eps_0$-far from convex with probability $1-o(1)$,
  for some positive constant $\eps_0$.
In the proof of the lemma we fix both the constant $\eps_0$ and our choice of $r=\Theta(n^{1/4})$.
(We remind the reader that $\rho$ and $\Dno$ both depend on the value of $r$.)

\begin{lemma}
There exist a real value $r=\Theta(n^{1/4})$ with $e^{r^2/2}\ge N/n$ and a positive constant $\eps_0$ such
  that a set $\SS\leftarrow \Dno$ is $\eps_0$-far from convex with probability at least $1-o(1)$.
\end{lemma}
\begin{proof}
We need the following claim but delay its proof to the end of the subsection:
\begin{claim}\label{claim1}
  There exist an $r=\Theta(n^{1/4})$ with $e^{r^2/2}\ge N/n$ and a constant $c\in (0,1/2)$ such that
   $$c < \rho(x) < 1-c,\quad\text{for all $x\in \left[\sqrt{n} - \red{10}, \sqrt{n} + \red{10}\right]$}. $$
\end{claim}

Let $K\subset [M]$ denote the set of all integers $k$ such that
  $[t_{k-1},t_k]\subseteq [\sqrt{n}-10,\sqrt{n}+10]$ (note that $K$ is a set of consecutive integers).
Observe that (1) the total probability mass of all shells $k\in K$ is at least
  $\Omega(1)$ (by Lemma \ref{johnstone}), and (2) the size $|K|$ is at least $\Omega(M)$ (which follows
  from (1) and the fact that all shells have the same probability mass).
\def\TT{\boldsymbol{T}}

Consider the following $1$-dimensional scenario.
We have $|K|$ intervals $[t_{k-1},t_k]$ and~draw a set $\TT$ by including each interval  independently
  with probability $\rho(t_k)$.
We prove the following~claim:

\begin{claim}\label{oror}
The random set $\TT$ satisfies the following property with probability at least $1-o(1)$:
For any interval $I \subseteq \R_{\geq 0}$, either $I$ contains $\Omega(M)$ intervals $[t_{k-1},t_k]$ that are not included in $\TT$,
  or $\overline{I}$ contains $\Omega(M)$ intervals $[t_{k-1},t_k]$ included in $\TT$.
\end{claim}
\begin{proof}

First note that it suffices to consider intervals $I\subseteq \cup_{k\in K} [t_{k-1},t_k]$ and
  moreover, we may further assume that both endpoints of $I$ come from endpoints of $[t_{k-1},t_k]$, $k\in K$.
(In other words, for a given outcome $T$ of $\TT$, if there exists an interval $I$ that violates the condition, i.e.,
  both $I$ and $\overline{I}$ contain fewer than $\Omega(M)$ intervals,
  then there is such an interval $I$ with both ends from end points of $[t_{k-1},t_k]$).
This assumption allows us to focus on $|K|^2\le M^2$ many possibilities for $I$ (as we will see below, our argument applies a union bound over these $K^2$ possibilities).

Given a candidate such interval $I$, we consider two cases.
If $I$ contains $\Omega(M)$ intervals $[t_{k-1},t_k]$, $k\in K$, then
  it follows from Claim \ref{claim1} and a Chernoff bound that $I$
  contains at least $\Omega(M)$ intervals not included in $\TT$ with probability $1-2^{-\Omega(M)}$.
On the other hand, if $\overline{I}$ contains $\Omega(M)$ intervals,
  then the same argument shows that $\overline{I}$ contains $\Omega(M)$ interals
  included in $\TT$ with probability $1-2^{-\Omega(M)}$.
The claim follows from a union bound over all the $|K|^2$ possibilities for $I$.
\end{proof}

We return to the $n$-dimensional setting and
  consider the intersection of $\SS\leftarrow \Dno$ with a ray starting from the origin.
Note that the intersection of the ray and any convex set is an interval on the ray.
As a result, Claim \ref{oror} shows that with probability at least $1-o(1)$ (over the draw of $\SS\leftarrow\Dno$),
  the intersection of any convex set with any ray either
  contains $\Omega(M)$ intervals $[t_{k-1},t_k]$ such that shell $k\in K$ is not included in $\SS$,
  or misses $\Omega(M)$ intervals $[t_{k-1},t_k]$ such that shell $k\in K$ is included in $\SS$.
Since by (1) above shells $k\in K$ together have $\Omega(1)$ probability mass under $\gaus$
  and each shell contains the same probability mass, we have that with probability $1-o(1)$,
  $\SS$ is $\eps_0$-far from any convex set for some constant $\eps_0>0$.
(A more formal argument can be given by performing integration using spherical coordinates
  and applying (\ref{tttt}).)
\end{proof}

\begin{proof}[Proof of Claim \ref{claim1}]
We start with the choice of $r$. Let
$$
\alpha=\sqrt{n}-\red{10}\quad\text{and}\quad \beta=\sqrt{n}+\red{10}.
$$
Let $\capp(t)$ denote the fractional surface area of the spherical cap $S^{n-1}\cap \{x:x_1\ge t\}$, i.e.,
$$
\capp(t)=\Pr_{\bx\leftarrow S^{n-1}} \big[\bx_1\ge t\big].
$$
So $\capp$ is a continuous, strictly decreasing function over $[0,1]$.
Since $\capp(0)=1/2$ and $\capp(1)=0$, there is a unique $r\in (0,\red{\alpha})$ such that
  $\red{\capp}(r/\alpha)=1/N=2^{-\sqrt{n}}$.
Below we show that $r=\Theta(n^{1/4})$ and fix it in the rest of the proof.
First recall the following explicit expression (see e.g. \cite{KOS:07}):
$$
\capp(t)=a_n\int_{t}^1\left(\sqrt{1-z^2}\right)^{n-3} dz,
$$
where $a_n = \Theta(n^{1/2})$ is a parameter that only depends on $n$.
Also recall the following inequalities from \cite{KOS:07} about $\capp(t)$:
\begin{equation}\label{KOSinequality}
\capp(t)\le e^{-nt^2/2},\quad \text{~for all $t \in [0,1]$}; \quad \quad
 \capp(t)\ge \Omega\left( t\cdot e^{-nt^2/2}\right),\quad \text{for $t=O(1/n^{1/4})$}.
\end{equation}
By our choice of $\alpha$ and the monotonicity of the cap function,
  this implies that $r=\Theta(n^{1/4})$ and
$$
1/N=\capp(r/\alpha)\ge \Omega(1/n^{1/4})\cdot e^{-n(r/\alpha)^2/2}\ge \Omega(1/n^{1/4})
\cdot e^{-(r^2/2)(1+O(1/\sqrt{n}))}=\Omega(1/n^{1/4})\cdot e^{-r^2/2}
$$
(using $r = \Theta(n^{1/4})$ for the last inequality), and thus, we have $e^{r^2/2}\ge N/n$.

Next, using the function $\capp$ we have the following expression for $\rho$:
\begin{equation} \label{eq:rhodef}
\rho(x) = \left(1 - \capp\left(\frac{r}{x}\right)\right)^N.
\end{equation}
As a side note, $\rho$ is continuous and thus, Lemma \ref{continuity} follows.
Since $\capp$ is strictly decreasing, we have that $\rho$ is strictly decreasing as well.
To finish the proof it suffices to show that there is a constant $c\in (0,1/2)$ such that
  $\rho(\alpha)<1-c$ and $\rho(\beta)\ge c$.
The first part is easy since
$$\rho(\alpha)=\left(1-1/N\right)^N\approx e^{-1} $$by our choice of $r$.
In the rest of the proof we show that
\begin{equation}\label{lefteq}
 \capp\left(\frac{r}{\beta}\right) \leq a \cdot\capp\left(\frac{r}{\alpha}\right)=\frac{a}{N},
\end{equation}
for some positive constant $a$.
It follows immediately that
$$\rho(\beta)=\left(1-\capp\left(\frac{r}{\beta}\right)\right)^N\ge
  \left(1-\frac{a}{N}\right)^N\ge \left(e^{-2a/ N }\right)^N=e^{-2a},
$$
using $1-x\ge e^{-2x}$ for $0\le x\ll 1$,
and this finishes the proof of the claim.

\begin{figure}[t]
\centering
\includegraphics[width=15cm]{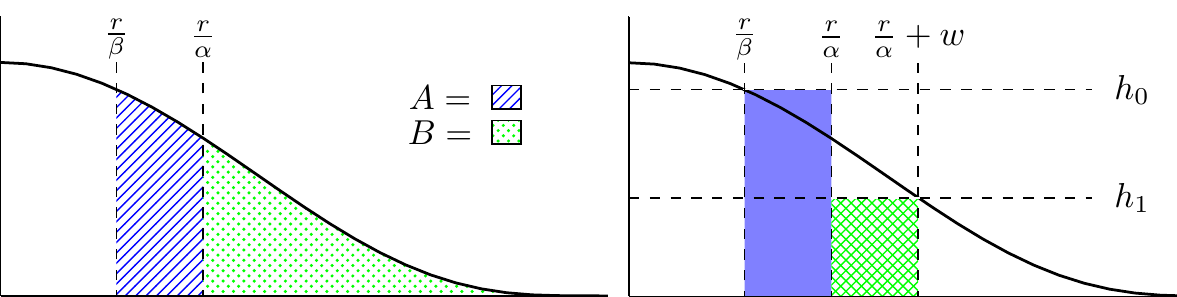}\vspace{0.3cm}
\caption{A plot of the integrand $(\sqrt{1-z^2})^{(n-3)}$. Area $A$ is $\capp(r/\beta)-\capp(r/\alpha)$ and area $B$ is $\capp(r/\alpha)$. The rectangles on the right are an upper bound of $A$ and a lower bound of $B$.}
\label{fig-rho}
\end{figure}

Finally we prove (\ref{lefteq}).
Let
$$
w=\frac{r}{\alpha}-\frac{r}{\beta}=\Theta\left(\frac{1}{n^{3/4}}\right)
$$
since $r=\Theta(n^{1/4})$.
Below we show that
\begin{equation}\label{riririr}
\int_{r/\beta}^{r/\alpha} \left(\sqrt{1-z^2}\right)^{n-3} dz
\le a'\cdot \int_{r/\alpha}^{r/\alpha \red{+}w} \left(\sqrt{1-z^2}\right)^{n-3} dz,
\end{equation}
for some positive constant $a'$.
It follows that
$$
\capp\left(\frac{r}{\beta}\right)-\capp\left(\frac{r}{\alpha}\right)
\le a'\cdot \capp\left(\frac{r}{\alpha}\right)
$$
and implies (\ref{lefteq}) by setting $a=a'+1$.
For (\ref{riririr}), note that the ratio of the $[r/\beta,r/\alpha]$-integration over
  the $[r/\alpha,r/\alpha\red{+}w]$-integration is at most
$$
\left(\frac{\sqrt{1-(r/\beta)^2}}{\sqrt{1-(r/\beta+2w)^2}}\right)^{n-3}
$$
as the length of the two intervals are the same and the function $(\sqrt{1-z^2})^{n-3}$
  is strictly decreasing. Figure~\ref{fig-rho} illustrates this calculation.
Let $\tau=r/\beta=\Theta(1/n^{1/4})$.
We can rewrite the~above~as
$$
\left(
\frac{1-\tau^2}{1-(\tau+2w)^2}
\right)^{(n-3)/2}
=\left(1+\frac{4\tau w+4w^2}{1-(\tau+2w)^2}\right)^{(n-3)/2}
=\left(1+O\left(\frac{1}{n}\right)\right)^{(n-3)/2}=O(1).
$$
This finishes the proof of the claim.
\end{proof}

\subsection{Distributions $\Eyes$ and $\Eno^*$ are close}\label{finalsec}

In the rest of the section we show that the total
  variation distance between $\Eyes$ and $\Eno^*$ is $o(1)$ and thus prove Lemma \ref{techlemma2}.
Let $z=(z_1,\ldots,z_q)$ be a sequence of $q$ points in $\R^n$.
We use $\Eyes(z)$ to denote the distribution of labeled samples from $\Eyes$,
  conditioning on the samples being $z$, i.e., $(z,\SS(z))$ with $\SS\leftarrow \Dyes$.
We let $\Eno^*(z)$ denote the distribution of labeled samples from $\Eno^*$,
  conditioning on the samples being $z$, i.e., $(z,\bb)$ where each $\bb_i$ is $1$ independently with probability
  $\rho(\|z_i\|)$.
Then
\begin{equation}\label{fuif}\dtv(\Eyes,\Eno^*)
=\E_{\zz\leftarrow (\gaus)^q} \Big[\dtv (\Eyes(\zz),\Eno^*(\zz))\Big].
\end{equation}

We split the proof of  Lemma \ref{techlemma2} into two steps.
We first introduce the notion of \emph{typical} sequences $z$ of $q$ points
  and show in this subsection that with probability $1-o(1)$, $\zz\leftarrow (\gaus)^q$ is typical.
In the next subsection we show that $\dtv(\Eyes(z),\Eno^*(z))$ is $o(1)$ when $z$ is typical.
It follows from (\ref{fuif}) that $\dtv(\Eyes,\Eno^*)$ is $o(1)$.
We start with the definition of typical sequences.

Given a point $z\in \R^n$, we are interested in the \emph{fraction} of points $y$
  (in terms of the area) in $S^{n-1}(r)$ such that $z\cdot y>r^2$.  This is because if any such point $y$ is sampled
  in the construction of $\SS\leftarrow \Dyes$, then $z\notin \SS$.
This is illustrated in Figure \ref{fig:pic}.
We refer to the set of such points $y$ as the \emph{\emph{(}spherical\emph{)} cap covered by $z$}
  and we write $\cover(z)$ to denote it.
(Note that $\cover(z)=\emptyset$ if $\|z\|\le r$.)

Given a subset $H$ of $S^{n-1}(r)$ (such as $\cover(z)$), we use
  $\fsa(H)$ to denote the fractional surface area of $H$ with respect to $S^{n-1}(r)$.
Using Figure \ref{fig:pic} and elementary geometry, we have the following connection between
  the fractional surface area of $\cover(z)$ and the cap function (for $S^{n-1}$):
\begin{equation} \label{eq:capcover}
\fsa\big(\cover(z)\big)=\capp\big(r/\|z\|\big).
\end{equation}

We are now ready to define typical sequences.

\begin{figure}[t!]
\centering
\includegraphics{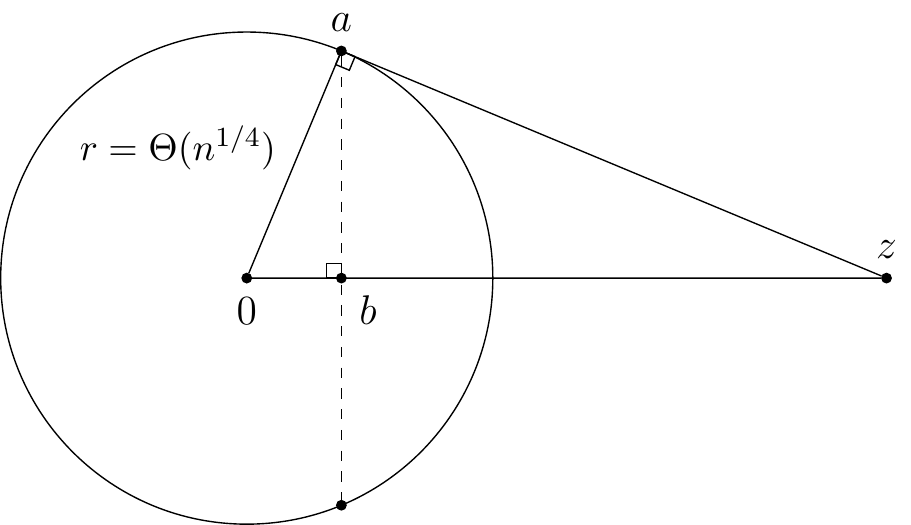}
\caption{The fractional surface area of $\cover(z)$, $\fsa(\cover(z))$, is the fraction of $S^{n-1}(r)$ to the right of the dashed line.
By similarity of triangles $0az$ and $0ba$, scaling down to the unit sphere, we get (\ref{eq:capcover}).}\label{fig:pic}
\end{figure}

\begin{definition}
We say a sequence $z=(z_1,\ldots,z_q)$ of $q$ points in $\R^n$ is \emph{typical} if
\begin{enumerate}
\item For every point $z_i$, we have
\begin{equation}\label{lulala}
\fsa\big(\cover(z_i)\big) \in \left[e^{-0.51\hspace{0.03cm} r^2 }, e^{-0.49\hspace{0.03cm}r^2}\right].
\end{equation}
\item For every $i \neq j$, we have
$$\fsa\big(\cover(z_i) \cap \cover(z_j)\big) \le
 e^{-0.96\hspace{0.03cm}r^2} .$$
\end{enumerate}
\end{definition}

The first condition of typicality essentially says that every $z_i$ is not too
  close to and not too far away from the origin (so that we have a relatively tight bound on the fractional surface area
  of the cap covered by $z_i$).
The second condition says that the caps covered by two points $z_i$ and $z_j$ have very little intersection.
We prove the following lemma:

\begin{lemma} \label{thm:most-z-are-typical}
$\zz\leftarrow (\gaus)^q$ is typical with probability at least \(1 - o(1)\).
\end{lemma}
\begin{proof}
We show that $\zz$ satisfies each of the two conditions with probability $1-o(1)$.
The lemma then follows from a union bound.

For the first condition, we let $c^*=0.001$ be a sufficiently small constant.
We have from Lemma \ref{johnstone} and a union bound that every
  $\zz_i$ satisfies $(1-c^*)\sqrt{n}\le \|\zz_i\|\le (1+c^*)\sqrt{n}$ with probability $1-o(1)$.
When this happens, we have (\ref{lulala}) for every $\zz_i$ using
  (\ref{KOSinequality}) and the upper bound of $\capp(t)\le e^{-nt^2/2}$.

For the second condition, we first note that the argument used in the first part implies that
$$
\E_{\bz_i\leftarrow\gaus}\Big[\fsa\big(\cover(\bz_i)\big)\Big]\le e^{-0.49\hspace{0.03cm}r^2}.
$$
Let $x_0$ be a fixed point in $S^{n-1}(r)$.
Viewing the fractional surface area as the following probability
$$
\fsa\big(\cover(z_i)\big) = \Pr_{\bx\leftarrow S^{n-1}(r)} \big[\bx\in \cover(z_i)\big],
$$
we have
\begin{align}
e^{-0.49\hspace{0.03cm}r^2}&\ge \E_{\bz_i\leftarrow \gaus}\Big[\fsa\big(\cover(\bz_i)\big)\Big] \label{eq:zzzz}\\
&=\E_{\bz_i}\Big[\Pr_{\bx\leftarrow S^{n-1}(r)} \big[\bx\in \cover(\zz_i)\big]\Big] \nonumber \\[0.5ex]
&=\Pr_{\bx,\zz_i}\big[\bx\in \cover(\zz_i)\big] =\Pr_{\zz_i}\big[x_0\in \cover(\zz_i)\big], \nonumber
\end{align}
where the last equation follows by sampling $\bx$ first and spherical and Gaussian symmetry.

Similarly we can express the fractional surface area of
  $\cover(z_i)\cap \cover(z_j)$ as
$$\fsa\big(\cover(z_i)\cap \cover(z_j)\big)=\Pr_{\bx\leftarrow S^{n-1}(r)} \big[\bx\in \cover(z_i)\ \text{and}\
  \bx\in \cover(z_j)\big].$$
We consider the expectation over $\zz_i$ and $\zz_j$ drawn independently from $\gaus$:
\begin{align*}
&\E_{\zz_i,\zz_j}\Big[\fsa\big(\cover(\zz_i)\cap \cover(\zz_j)\big)\Big]\\
&=\E_{\zz_i,\zz_j}\Big[\Pr_{\bx\leftarrow S^{n-1}(r)} \big[\bx\in \cover(\zz_i)\ \text{and}\
  \bx\in \cover(\zz_j)\big]\Big]\\
&=\Pr_{\bx,\zz_i,\zz_j}\big[\bx\in \cover(\zz_i)\ \text{and}\
  \bx\in \cover(\zz_j)\big]=\Pr_{\zz_i}\big[x_0\in \cover(\zz_i)\big]\cdot \Pr_{\zz_j}\big[x_0\in \cover(\zz_j)\big],
\end{align*}
where the last equation follows by sampling $\bx$ first, independence of $\zz_i$ and $\zz_j$, and symmetry.

By (\ref{eq:zzzz}), the expectation of $\fsa(\cover(\zz_i)\cap \cover(\zz_j))$ is
  at most $\smash{e^{-0.98\hspace{0.03cm}r^2}}$, and hence by Markov's inequality,
the probability of it being at least $\smash{e^{-0.96\hspace{0.03cm}r^2}}$
  is at most $\smash{e^{-0.02\hspace{0.03cm}r^2}}$.
Using $\smash{e^{r^2}\ge (N/n)^2}$ and a union bound, the probability of one of the pairs having
  the $\fsa$ at least $\smash{e^{-0.96\hspace{0.03cm}r^2}}$ is at most
$$
q^2 \cdot e^{-0.02r^2}\le 2^{0.02\hspace{0.03cm}\sqrt{n}}\cdot (n/N)^{0.04}=o(1),
$$
since $q=2^{0.01\sqrt{n}}$ and $N=2^{\sqrt{n}}$.
This finishes the proof of the lemma.
\end{proof}

We prove the following lemma in Section \ref{hehefinal} to finish the proof of Lemma \ref{techlemma2}.

\begin{lemma} \label{thm:typical-z-are-good}
For every typical sequence $z$ of $q$ points, we have \(d_{\text{TV}}\big(\Eyes(z), \Eno^*(z)\big) = o(1)\).
\end{lemma}

\subsection{Proof of Lemma~\ref{thm:typical-z-are-good}}\label{hehefinal}

Fix a typical $z=(z_1,\ldots,z_q)$.
Our goal is to show that the total variation distance of $\Eyes(z)$ and $\Eno^*(z)$ is $o(1)$.
To this end, we define a distribution $\calF$ over pairs $(\bb,\dd)$ of strings in $\{0,1\}^q$ (as a coupling
  of $\Eyes(z)$ and $\Eno^*(z)$), where the
  marginal distribution of $\bb$ as $(\bb,\dd)\leftarrow \calF$ is the same as $\Eyes(z)$ and
  the marginal distribution of $\dd$ is the same as $\red{\Eno^\ast}(z)$.
Our goal follows by establishing
\begin{equation}\label{maineq}
\Prx_{(\bb,\dd)\leftarrow \calF} \big[\bb\ne \dd\big]=o(1).
\end{equation}

To define $\calF$, we use $\bM$ to denote the $q\times N$ $\{0,1\}$-valued random matrix
  derived from $z$ and $\SS\leftarrow \Dyes$ (recall that $\SS$ is the intersection of
  $N$ random halfspaces $\bh_j$, $j\in [N]$):
the $(i,j)$th entry $\bM_{i,j}$ of $\bM$ is $1$ if $\bh_j(z_i)=1$ (i.e., $z_i\in \bh_j$)
  and is $0$ otherwise.
We use $\bM_{i,*}$ to denote the $i$th row of $\bM$,
  $\bM_{*,j}$ to denote the $j$th column of $\bM$,
  and $\bM^{(i)}$ to denote the $i\times N$ sub-matrix of $\bM$  that consists of the first $i$ rows of $\bM$.
(We note that $\bM$ is derived from $\SS$ and they are defined over the same probability space.
So we may consider the (conditional) distribution of $\SS\leftarrow \Dyes$
  conditioning on an event involving $\bM$, and we may consider the conditional  distribution of $\bM$
  conditioning on an event involving $\SS$.)

\def\barbx{\mathbf{b}}
\def\barby{\mathbf{d}}

We now define the distribution $\calF$.
A pair $(\bb,\dd)\leftarrow \calF$ is drawn using the following randomized procedure.
The procedure has $q$ rounds and generates the $i$th bits $\bb_i$ and $\dd_i$ in the $i$th round:
\begin{flushleft}\begin{enumerate}
\item In the first round, we draw a random real number $\br_1$ from $[0,1]$ uniformly at random.
We set $\bb_1=1$ if $\br_1\le \Pr_{\SS\leftarrow \Dyes} [\SS(z_1)=1]$ and set $\bb_1=0$ otherwise.
We then set $\barby_1=1$ if $\br_1\le \rho(\|z_1\|)$ and set $\barby_1=0$ otherwise.
(Note that for the first round, the two thresholds are indeed the same
  so we always have $\barbx_1=\barby_1$.)
At the end of the first round, we also draw a row vector $\bN_{1,*}$ according to the
  distribution of $\bM_{1,*}$ conditioning on $\SS(z_1)=\barbx_1$.

\item In the $i$th round, for $i$ from $2$ to $q$, we draw a random real number $\br_i$ from $[0,1]$ uniformly at random.
We set $\barbx_i=1$ if we have $$\br_i\le \Ppr_{\SS\leftarrow \Dyes}\Big[\SS(z_i)=1\hspace{0.06cm}\big|\hspace{0.06cm}\bM^{(i-1)}=\bN^{(i-1)}\Big]$$
and set $\barbx_i=0$ otherwise.
We then set $\barby_i=1$ if $\br_i\le \rho(\|z_i\|)$ and set $\barby_i=0$ otherwise.
At the end of the $i$th round, we also draw a row vector $\bN_{i,*}$ according to the
  distribution of $\bM_{i,*}$ conditioning on $\bM^{(i-1)}=\bN^{(i-1)}$ and $\SS(z_i)=\barbx_i$.
\end{enumerate}\end{flushleft}
It is clear that the marginal distributions of $\bb$ and $\dd$, as $(\bb,\dd)\leftarrow \calF$,
  are $\Eyes$ and $\Eno^*$ respectively.

To prove (\ref{maineq}), we introduce the following notion of \emph{nice} and \emph{bad} matrices.
\begin{definition}
Let $M$ be an $i\times N$ $\{0,1\}$-valued matrix for some $i\in [q]$.
We say $M$ is \emph{nice} if
\begin{enumerate}
\item $M$ has at most $\sqrt{N}$ many $0$-entries; and\vspace{-0.1cm}
\item Each column of $M$ has at most one 0-entry.
\end{enumerate}
We say $M$ is \emph{bad} otherwise.
\end{definition}

We prove the following two lemmas and use them to prove (\ref{maineq}).

\begin{lemma}\label{lem1}
$\Pr_{\SS\leftarrow \Dyes }\big[\bM\ \text{is bad}\big]=o(1/q)$.
\end{lemma}

Note that when $\bM$ is nice, we have by definition that $\bM^{(i)}$ is also nice for every $i\in [q]$.

\begin{lemma}\label{lem2}
For any nice $(i-1)\times N$  $\{0,1\}$-valued matrix $M^{(i-1)}$, we have
\begin{equation}\label{huha}
\Prx_{\SS\leftarrow \Dyes}\Big[\SS(z_i)=1\hspace{0.06cm}\big|\hspace{0.06cm}
\bM^{(i-1)}=M^{(i-1)}\Big]=\rho(\|z_i\|)\pm o(1/q).
\end{equation}
\end{lemma}

Before proving Lemma \ref{lem1} and \ref{lem2},
  we first use them to prove (\ref{maineq}).
Let $\bI_i$ denote the indicator random variable that is $1$ if $(\barbx,\barby)\leftarrow \calE$ has
  $\barbx_i\ne \barby_i$ and is $0$ otherwise, for each $i\in [q]$.
Then (\ref{maineq}) can be bounded from above by
  $\sum_{i\in [q]} \Pr[\bI_i=1]$.
To bound each $\Pr[\bI_i=1]$ we split the event into
$$
\sum_{M^{(i-1)}} \Pr\big[\bN^{(i-1)}=M^{(i-1)}\big]\cdot
  \Pr\big[\bI_i=1\hspace{0.06cm}|\hspace{0.06cm}\bN^{(i-1)}=M^{(i-1)}\big],
$$
where the sum is over all $(i-1)\times N$ $\{0,1\}$-valued matrices $M^{(i-1)}$,
and further split the sum into two sums over nice and
  bad matrices $M^{(i-1)}$.
As  $\bN^{(i-1)}$ has the same
  distribution as $\bM^{(i-1)}$, it follows from Lemma \ref{lem1} (and the fact that $\bM$ is bad when
  $\bM^{(i-1)}$ is bad) that the sum over bad $M^{(i-1)}$ is
  at most $o(1/q)$.
On the other hand, it follows from Lemma \ref{lem2} that the sum over nice $ M^{(i-1)}$
  is $o(1/q)$.
As a result, we have $\Pr[\bI_i=1]=o(1/q)$ and thus, $\sum_{i\in [q]} \Pr[\bI_i=1]=o(1)$.

We prove Lemmas \ref{lem1} and \ref{lem2} in the rest of the section.

\begin{proof}[Proof of Lemma \ref{lem1}]
We show that the probability of $\bM$ violating each of the two conditions in the definition
  of nice matrices is $o(1/q)$. The lemma
  then follows by a union bound.

For the first condition, since $z$ is typical the probability of $\bM_{i,j}=0$ is
$$
\fsa\big(\cover(z_i)\big)\le e^{-0.49\hspace{0.03cm}r^2}.
$$
By linearity of expectation, the expected number of $0$-entries in $\bM$ is at most
$$qN\cdot e^{-0.49\hspace{0.03cm}r^2} = o(\sqrt{N}/q),$$
using $e^{r^2/2}\ge N/n$, $N=2^{\sqrt{n}}$ and $q=2^{0.01\sqrt{n}}$.
It follows directly from Markov's inequality that the probability of
  $\bM$ having more than $\sqrt{N}$ many $0$-entries is $o(1/q)$.

For the second condition, again since $z$ is typical, the probability of $\bM_{i,j}=\bM_{i',j}=1$ is
$$
\fsa\big(\cover(z_i)\cap \cover(z_i')\big)\le e^{-0.96\hspace{0.03cm}r^2}.
$$
By a union bound, the probability of $\bM_{i,j}=\bM_{i',j}=1$ for some $i,i',j$ is at most
$$q^2N\cdot e^{\red{-0.96r^2}\hspace{0.03cm}}=o(1/q).$$
This finishes the proof of the lemma.
\end{proof}

Finally we prove Lemma \ref{lem2}.
Fix a nice $(i-1)\times N$ matrix $M$ (we henceforth omit the superscript $(i-1)$ since the
  number of rows of $M$ is fixed to be $i-1$).
Recall that $\SS(z_i)=1$ if and only if $\bh_j(z_i)=1$ for all $j\in [N]$.
As a result, we have
$$
\Ppr_{\SS\leftarrow \Dyes}\Big[\SS(z_i)=1\hspace{0.06cm}\big|\hspace{0.06cm}
\bM^{(i-1)}=M\Big]
=\prod_{j\in [N]} \hspace{0.05cm}\Ppr_{\bh_j}\Big[\bh_j(z_i)=1\hspace{0.06cm}\big|\hspace{0.06cm}
  \bM^{(i-1)}_{*,j}=M_{*,j}\Big].
$$
On the other hand, letting $\tau=\fsa(\cover(z_i))=\capp(r/\|z_i\|)$, we have
  $\rho(\|z_i\|)=(1-\tau)^N$.

In the next two claims we compare
$$
\Ppr_{\bh_j}\Big[\bh_j(z_i)=1\hspace{0.06cm}\big|\hspace{0.06cm}
  \bM^{(i-1)}_{*,j}=M_{*,j}\Big]
$$
with $1-\tau$ for each $j\in [N]$ and show that they are very close.
The first claim works on $j\in [N]$ with no $0$-entry in $M_{*,j}$ and
  the second claim works on $j\in [N]$ with one $0$-entry in $M_{*,j}$.
(These two possibilities cover all $j\in [N]$ since the matrix $M$ is nice.)
Below we omit $\bM^{(i-1)}_{*,j}$ in writing the conditional probabilities.

\begin{claim}\label{lem-goodj}
For each $j\in [N]$ with no $0$-entry in the $j$th column $M_{*,j}$, we have
$$\Prx_{\bh_j}\Big[\bh_j(z_i)=1\hspace{0.06cm}\big|\hspace{0.06cm} M_{*,j}\Big]
 = (1-\tau)\left(1\pm \frac{o(1)}{qN}\right).$$
\end{claim}
\begin{proof}
Let
  $\delta$ be the probability of $\bh_j(z_i)=0$ conditioning on $M_{*,j}$ (which is all-$1$).
Then
$$
\delta=\frac{\fsa\left(\cover(z_i)-\bigcup_{j<i} \cover(z_j)\right)}{
1-\fsa\left(\bigcup_{j<i} \cover(z_j)\right)}.
$$
Using $e^{-0.51\hspace{0.03cm}r^2}\le \fsa(\cover(z_j))\le e^{-0.49\hspace{0.03cm}r^2}$ and $\fsa(\cover(z_i)\cap \cover(z_j))\le e^{-0.96
\hspace{0.03cm}r^2}$, we have
\begin{align*}
\delta &\leq \frac{\tau}{1-q\cdot e^{-0.49\hspace{0.03cm}r^2}}<\tau(1+2q\cdot e^{-0.49\hspace{0.03cm}r^2})
=\tau+2\tau q\cdot e^{-0.49\hspace{0.03cm}r^2}.
\end{align*}
Using $\tau\le e^{-0.49\hspace{0.03cm}r^2}$ and $e^{r^2/2}\ge N/n$, we have
$$
1-\delta\ge 1-\tau -2\tau q\cdot e^{-0.49\hspace{0.03cm}r^2}
\ge 1-\tau-o\big(1/(qN)\big)\ge (1-\tau)\big(1-o(1/(qN))\big).
$$
On the other hand, we have

$\delta \geq \tau -q\cdot e^{-0.96\hspace{0.03cm}r^2}$ and thus,
$$
1-\delta\le 1-\tau+q\cdot e^{-0.96\hspace{0.03cm}r^2}\le 1-\tau+o\big(1/(qN)\big)
=(1-\tau)\big(1+o(1/(qN))\big).
$$
This finishes the proof of the claim.
\end{proof}
\begin{claim}\label{lem-badj}
For each $j\in [N]$ with one $0$-entry in the $j$th column $M_{*,j}$, we have
$$\Pr_{\bh_j}\Big[\bh_j(z_i)=1\hspace{0.06cm}\big|\hspace{0.06cm} M_{*,j}\Big] \geq 1-O\big(e^{-0.45\hspace{0.03cm}r^2}\big).$$
\end{claim}
\begin{proof}
Let $i'$ be the point with $M_{i',j}=1$ and $\delta$
  be the conditional probability of $\bh_j(z_i)=0$. Then
$$
\delta\le \frac{\fsa\big(\cover(z_i)\cap \cover(z_{i'})\big)}{\fsa\left(\cover(z_i')-
\bigcup_{j<i:\hspace{0.03cm}j\ne i'} \cover(z_j)\right)}
\le \frac{e^{-0.96\hspace{0.03cm}r^2}}{e^{-0.51\hspace{0.03cm}r^2}-q\cdot e^{-0.96\hspace{0.03cm}r^2}}
=O\big(e^{-0.45\hspace{0.03cm}r^2}\big),
$$
by our choice of $q$.
This finishes the proof of the claim.
\end{proof}

We combine the two claims to prove Lemma \ref{lem2}.

\begin{proof}[Proof of Lemma \ref{lem2}]
Let $h$ be the number of $0$-entries in $M$.
We have $h\le \sqrt{N}$ since $M$ is nice.
By Claims~\ref{lem-goodj}, the conditional probability of $\SS(z_i)=1$ is at most
\begin{align*}
\left((1-\tau)\left(1+o\left(\frac{1}{qN}\right)\right)\right)^{N-h}
&=\rho(\|z_i\|)\cdot \frac{1}{(1-\tau)^h}\cdot \left(1+o\left(\frac{1}{qN}\right)\right)^{N-h} \\
&\le \rho(\|z_i\|) \cdot (1+2\tau)^h\cdot \left(1+o\left(\frac{1}{qN}\right)\right)^{N}\\[0.3ex]
&\le \rho(\|z_i\|)\cdot \exp\big(2\tau h+ o(1/q)\big)\\[0.6ex]
&=\rho(\|z_i\|)\cdot \exp\big(o(1/q)\big)=\rho(\|z_i\|)+o(1/q).
\end{align*}
Similarly, the conditional probability of $\SS(z_i)=1$ is at least
\begin{align*}
&\left((1-\tau)\left(1-o\left(\frac{1}{qN}\right)\right)\right)^{N-h}\left(1-O\left(e^{-0.45\hspace{0.03cm}r^2}\right)\right)^{h}\\
&\hspace{1cm}\ge \rho(\|z_i\|)\cdot \left(1-o\left(\frac{1}{qN}\right)\right)^{N-h}
\left(1-O\left(e^{-0.45\hspace{0.03cm}r^2}\right)\right)^{h}\\[0.6ex]
&\hspace{1cm}\ge \rho(\|z_i\|)\cdot \big(1-o(1/q)\big)\ge \rho(\|z_i\|)-o(1/q).
\end{align*}
This finishes the proof of the lemma.
\end{proof}

%% file: 1slb.tex

\section{One-sided lower bound} \label{sec:1slb}

We recall Theorem~\ref{thm:1slb}:
\begin{reptheorem}{thm:1slb}
\red{Any one-sided sample-based algorithm that is an $\eps$-tester for convexity over $\normal^n$
  for some $\eps<1/2$
  must use $2^{\Omega(n)}$ samples.}
\end{reptheorem}

We say a finite set $\{x^1,\dots,x^M\} \subset \R^n$ is  \emph{shattered} by $\calC_\convex$ if for every $(b_1,\dots,b_M) \in \zo^M$ there is a convex set $C \in \calC_\convex$ such that $C(x^i)=b_i$ for all $i \in [M].$
Theorem~\ref{thm:1slb} follows from the following lemma:

\begin{lemma} \label{lem:any-labeling}There is an absolute constant $c>0$ such that for $M=2^{c n}$, it holds that
\[
\Prx_{\bx^i\leftarrow\normal^n}\big[\{\bx^1,\dots,\bx^M\} \text{~is shattered by~}\calC_\convex\hspace{0.03cm}\big] \geq 1 - o(1).
\]
\end{lemma}

\begin{proof}[Proof of Theorem~\ref{thm:1slb} using Lemma~\ref{lem:any-labeling}.]
Suppose that $A$ were a one-sided sample-based algorithm for $\eps$-testing $\calC_\convex$ using at most $M$ samples.  Fix a set $\targetset$ that is $\eps$-far from $\calC_\convex$ to be the unknown target subset of $\R^n$ that is being tested.\footnote{An example of such a subset $S$ is as follows (we define it as
  a function $S:\R^n\rightarrow \{0,1\}$):  Given an odd integer $N > (1/2 - \eps)\inv - 1$, let $-\infty = \tau_0 < \tau_1 < \cdots < \tau_N < \tau_{N+1}=+\infty$ be values such that
$
\Pr_{\bz \leftarrow \normal}[\bz \leq \tau_i] =  i/( {N+1}),
$
and let $\targetset: \R^n \to \{0,1\}$ be the function defined by $\targetset(x_1,\dots,x_n)=\mathbf{1}[i$ is even$]$, where $i \in \{0,\dots,N\}$ is the unique value such that $\tau_i \leq x_1 < \tau_{i+1}.$  Fix any $z = (z_2,\dots,z_n) \in \R^{n-1}$ and we let $\targetset_z: \R \to \{0,1\}$ be the function defined as $\targetset_z(x_1)=\targetset(x_1,z_2,\dots,z_n)$.  An easy argument gives that $\targetset_z$ is $\left(1/2 - { 1 /({N+1})}\right)$-far (and hence $\eps$-far) from every convex subset of $\R$, and it follows by averaging (using the fact that the restriction of any convex subset of $\R^n$ to a line is a convex subset of $\R$) that $\targetset$ is $\eps$-far from $\calC_\convex.$}
Since $\targetset$ is $\eps$-far from convex, it must be the case that
\begin{equation}\label{haui}
\Prx_{\bx^i\leftarrow \normal^n}\big[A \text{\ rejects when run on $(\bx^1,\targetset(\bx^1)),\dots,
(\bx^M,\targetset(\bx^M))$}\big] \geq 2/3.
\end{equation}
But  Lemma~\ref{lem:any-labeling} together with the one-sidedness of $A$ imply that
\begin{align*}
&\Prx_{\bx^i\leftarrow\normal^n}\big[\hspace{0.03cm}\text{for any $(b^1,\dots,b^M) \in \zo^M$,
$A$ rejects when run on~$(\bx^1,b^1),\dots,
(\bx^M,b^M)$}\hspace{0.02cm}\big] \leq o(1),
\end{align*}
as $A$ can only reject if the labeled samples are not consistent with any convex set,
  which implies that $A$ cannot reject when $\{\bx^1,\ldots,\bx^M\}$ is shattered by $\calC_\convex$.
  This contradicts with (\ref{haui}).\end{proof}

In the next subsection we prove Lemma~\ref{lem:any-labeling} for $c=1/500.$

\subsection{Proof of Lemma~\ref{lem:any-labeling}}

Let $M=2^{cn}$ with $c=1/500$.
We prove the following lemma:

\begin{lemma} \label{lem:rand-all-on-hull}
For $\bx^1,\dots,\bx^M$ drawn independently from $\normal^n$, with probability $1-o(1)$ it is the case that for all $i \in [M],$ no $\bx^i$ lies in $\Conv(\{\bx^j:j \in [M] \setminus i\}).$
\end{lemma}

If $\bx^1,\dots,\bx^M$ are such that no $\bx^i$ lies in $\Conv(\{\bx^j:j \in [M] \setminus i\})$, then given any $(b^1,\dots,b^M)$,
  by taking $C=\Conv(\{\bx^i: b^i = 1\})$ we see that there is a convex set $C$ such that $C(\bx^i)=b^i$ for all $i \in [M]$.
Thus to establish Lemma~\ref{lem:any-labeling} it suffices to prove Lemma~\ref{lem:rand-all-on-hull}.

To prove Lemma~\ref{lem:rand-all-on-hull}, it suffices to show that for each fixed $j \in [M]$ we have
\begin{equation}
\Prx_{\bx^i\leftarrow \normal^n}\big[\hspace{0.03cm}\bx^j \in \Conv(\{\bx^k:k \in [M]\setminus \{j\}\})\big] \leq M^{-2}
\label{eq:A}
\end{equation}
since given this a union bound implies that
\[
\Prx_{\bx^i\leftarrow \normal^n}\big[\hspace{0.03cm}\text{for some $j \in [M]$, }\bx^j \text{~lies in~} \Conv(\{\bx^k :k \in [M]\setminus \{j\}\})\big] \leq M^{-1}=o(1).
\]
By symmetry, to establish (\ref{eq:A}) it suffices to show that
\begin{equation} \label{eq:B}
\Prx_{\bx^i \leftarrow \normal^n}\big[\bx^M \in \Conv(\{\bx^1,\dots,\bx^{M-1}\})\big] \leq M^{-2}.
\end{equation}

In turn (\ref{eq:B}) follows from the following inequalities
  ($v\in \R^n$ is a fixed unit vector in the second)
\begin{equation}\label{twoinequalities}
\Prx_{\bx \leftarrow \normal^n}\big[\|\bx\| \leq \sqrt{n}/10\big] < {\frac 1 2} M^{-2}\quad\text{and}\quad
\Pr_{\bx \leftarrow \normal^n}\big[\bx \cdot v \geq \sqrt{n}/10\big] < {\frac 1 2}M^{-3}.
\end{equation}
The first inequality follows directly from Lemma \ref{lem:johnstone} using $c=1/500$.
For the second, by the spherical symmetry of $\normal^n$ we may take $v=(1,0,\dots,0).$
Recall the standard Gaussian tail bound $$\Pr_{\bz \leftarrow \normal}\big[\bz\ge t\big] \leq e^{-t^2/2}$$ for $t\geq 0$. This gives us that
$$\Pr_{\bx \leftarrow \normal^n}\big[\bx \cdot v \geq \sqrt{n}/10\big] \leq e^{-n/200} < {\frac 1 2} M^{-3},$$ again using that $M=2^{cn}$ and $c=1/500.$

Finally,
  to see that (\ref{eq:B}) follows from (\ref{twoinequalities}), we observe first that by the first inequality we may assume that $\|\bx^M\| > \sqrt{n}/10$ (at the cost of failure probability at most $ M^{-2}/2$ towards (\ref{eq:B})); fix~any such outcome $x^M$ of $\bx^M.$  By a union bound over $\bx^1,\dots,\bx^{M-1}$ and the second inequality, we have
\[
\Prx_{\bx^i\leftarrow\normal^n}\left[\text{\hspace{0.03cm}any $i \in [M-1]$ has $\bx^i \cdot {\frac {x^M}{\|x^M\|}} \geq \sqrt{n}/10$}\right] <  {\frac 1 2} M^{-2}.
\]
 But if every $\bx^i$ has $\bx^i \cdot ({x^M}/{\|x^M\|})<\sqrt{n}/10 < \|x^M\|$, then $x^M\notin\Conv(\{\bx^1,\dots,\bx^{M-1}\}).$

%% file: 2sub.tex

\section{Two-sided upper bound} \label{sec:2sub}

Recall Theorem \ref{thm:2sub}:
\begin{reptheorem}{thm:2sub}
For any $\eps > 0$, there is a two-sided sample-based $\eps$-tester for convexity over
   $\normal^n$ using $n^{O(\sqrt{n}/\eps^2)}$ samples.
\end{reptheorem}

We begin by recalling some definitions from learning theory.  Let $\calC$ be a class of subsets of $\R^n$ (such as $\calC_\convex$).
We say an algorithm learns $\calC$ to error $\eps$ with confidence $1-\delta$ under $\normal^n$ if,
  given a set of labeled samples $(\bx,S(\bx))$ from an unknown set $S\in \calC$ with $\bx$'s drawn
  independently from $\normal^n$, the algorithm outputs with probability at least $1-\delta$
  a hypothesis set $H\subseteq \R^n$~with $\Vol(S\bigtriangleup H)\le \eps$.
We say it is a \emph{proper} learning algorithm if
  it always outputs a hypothesis $H$ that belongs to $\calC$.
Next we recall the main algorithmic result of \cite{KOS:07}:

\begin{theorem} [Theorem~5 of \cite{KOS:07}] \label{thm:KOSlearn}
There is an algorithm $A$ that learns the class $\calC_\convex$ of all convex subsets of $\R^n$ to error $\eps$ with confidence $1-\delta$ under $\normal^n$ using
$$n^{O(\sqrt{n}/\eps^2)} \cdot \log(1/\delta)$$ samples\footnote{Theorem~5 as stated in \cite{KOS:07} gives a sample complexity upper bound of
$\smash{n^{O(\sqrt{n}/\eps^4)}}$ for \emph{agnostic} learning, but inspection of the proof gives the theorem as stated here, with an  upper bound of $\smash{n^{O(\sqrt{n}/\eps^2)}}$  for non-agnostic learning.}
drawn  from $\normal^n$.
\end{theorem}

Next we recall the result of Goldreich, Goldwasser and Ron which relates proper learnability of a class $\calC$ to the testability of $\calC$.

\begin{theorem} [Proposition~3.1.1 of \cite{GGR98}, adapted to our context] \label{thm:GGRlearntest}
Let $\calC$ be a class of subsets of $\R^n$ that has a proper learning algorithm $A$ which uses $m_A(n,\eps,\delta)$ samples from $\normal^n$ to learn $\calC$ to error $\eps$ with confidence $1-\delta$.   Then there is a property testing algorithm $A_{\mathrm{test}}$ for $\calC$ under the distribution $\normal^n$ that uses
\[
m_A\big(n,\eps/2,\delta/2\big) + O\big(\log(1/\delta)/\eps\big)
\]
samples drawn from $\normal^n.$
\end{theorem}

By Theorem~\ref{thm:GGRlearntest}, to obtain Theorem~\ref{thm:2sub} it suffices to have a \emph{proper} learning analogue of~Theorem \ref{thm:KOSlearn}.  We establish the required result, as a corollary of Theorem~\ref{thm:KOSlearn}, in the next subsection:

\begin{corollary} \label{cor:KOSlearn}
There is a proper learning algorithm $A'$ for the class $\calC_\convex$ of all convex subsets of $\R^n$ that uses $n^{O(\sqrt{n}/\eps^2)} \cdot \log(1/\delta)$ samples from $\normal^n$ to learn to error $\eps$ with confidence $1-\delta$.
\end{corollary}

We remark that while algorithm~$A$ from Theorem~\ref{thm:KOSlearn} runs in time $n^{O(\sqrt{n}/\eps^2)}$ and uses
$n^{O(\sqrt{n}/\eps^2)}$ samples, the algorithm $A'$ of Corollary~\ref{cor:KOSlearn} presented below has a much larger running time (at least $(n/\eps)^{O(n)}$); however, its sample complexity is essentially no larger than that of algorithm~$A$.

\subsection{Proof of Corollary~\ref{cor:KOSlearn}}

The idea behind the proof of Corollary~\ref{cor:KOSlearn} is simple.  Let $\targetset \subseteq \R^n$ be the unknown target convex set that is to be learned. Algorithm $A'$ first runs algorithm $A$ with error parameter \blue{$\eps/5$} and confidence parameter $\delta/2$ to obtain, with probability  $1-(\delta/2)$, a hypothesis $H\subseteq \R^n$ with $\Vol(H \bigtriangleup \targetset) \leq \eps/\blue{5}.$

In the rest of the algorithm we find with high probability a convex set $C^*$ with
  $\Vol(H\bigtriangleup C^*)\le \blue{4\eps/5}$ and thus, we have $\Vol(S\bigtriangleup C^*)\le
  \blue{\eps/5+4\eps/5=\eps}$.
(Note that this part of the algorithm does not require any labeled samples
  $(\bx,S(\bx))$ from the oracle for $S$.)

\def\Ccover{\calC_{\mathrm{cover}}}

For this purpose let $\Ccover\subset \calC_\convex$ be a \emph{finite $\blue{(\eps/5)}$-cover} of $\calC_\convex$. (We show in Corollary \ref{cor:epscover} below that there is an algorithm the finds a finite $(\eps/5)$-cover of $\calC_\convex$.)
Next, the algorithm $A'$ enumerates over all elements $C \in \Ccover$ and for each such $C$ uses random sampling from $\normal^n$ to estimate $\Vol(H \bigtriangleup C)$ to within an additive error of $\eps/5$, with success probability $1-\delta/(2|\Ccover|)$ for each $C$.  (Note that this does not require any labeled samples $(\bx,\targetset(\bx))$ from the oracle for $\targetset$, since $A'$ can generate its own draws from $\normal^n$ and for each such  $\bx$ it can compute $H(\bx)$ and $C(\bx)$ on its own.)  $A'$ outputs the $C^* \in \Ccover$ for which the estimate of $\Vol(H \bigtriangleup C^*)$ is smallest.

The fact that this works follows a standard argument.  Since $$\Vol(H \bigtriangleup \targetset) \leq \eps/5\quad \text{and}\quad \Vol(\targetset \bigtriangleup C') \leq \eps/5$$ for some set $C'  \in \Ccover$, it holds that $\Vol(H \bigtriangleup C') \leq 2\eps/5$ and hence the estimate of $\Vol(H \bigtriangleup C')$ will be at most $3\eps/5$.  Thus the element $C^*$ of $\Ccover$ that is selected will have its estimated value of $\Vol(H \bigtriangleup C^\ast)$ being at most $3\eps/5,$ which implies that its actual value of $\Vol(H \bigtriangleup C^\ast)$ will be at most $4\eps/5$ (since each estimate is within $\pm \eps/5$ of the true value).

Given the above analysis, to finish the proof of Corollary~\ref{cor:KOSlearn} it suffices to establish the following corollary of structural results proved in Sections \ref{sec:structural} and \ref{sec:setup},
which shows that indeed it is possible for $A'$ to enumerate over the elements of $\Ccover$ as described above:

\begin{corollary}\label{cor:epscover}
\hspace{-0.03cm}There is an algorithm that, on inputs $\eps$ and $n$, outputs a finite $\eps$-cover of $\calC_\convex$.
\end{corollary}

\begin{proof}
We recall the material and parameter settings from Section~\ref{sec:setup}.
Since every convex set in $\R^n$ is $(\e/4)$-close to a set in $\calC'_{\convex}$,
it suffices to describe a finite family $\calC$ of convex sets $C_1,C_2,\dots$ such that every $C \in \calC'_{\convex}$ is
$(3\e/4)$-close to some $C_i$ in $\calC$.
We claim that
\[
\calC = \big\{\hspace{0.02cm}\Conv(\cup_{\Cube \in Q} \Cube) \mid Q \subseteq \CubeSet\hspace{0.03cm}\big\}
\]
is such a family.  To see this, fix any convex body $C \in \calC'_{\convex}$.
Let $$Q_C = \big\{\hspace{0.02cm}\Cube \in \CubeSet \mid \Cube \subseteq C\hspace{0.02cm}\big\},$$ the set of cubes that are entirely contained in $C$. Note that $\Conv(Q_C)$ is a subset of $C$. If a $\Cube$ contains at least one point in $C$ and at least one point outside $C$, then every point in~$\Cube$ has distance at most $\ell\sqrt{n}$ from the boundary of $C$ (since any two points in a given $\Cube$ have distance at most $\ell\sqrt{n}$). Thus, the missing volume $C \setminus \Conv(Q_C)$ is completely contained in $\del C + \Ball(\ell\sqrt{n})$, whose Gaussian volume, by Theorem~\ref{thm:surfacevolume}, is at most $20\hspace{0.03cm}
n^{5/8}\hspace{0.03cm} n'\sqrt{\ell \sqrt{n}} \ll 3\eps/4.$
\end{proof}

%% file: lemmas.tex

\section{Proof of Lemmas~\ref{lem:noball}, \ref{lem:small}, and \ref{lem:large}}  \label{ap:lemmas}
\begin{lemma}\label{lem:noball}
If $C \subset \R^n$ is convex and contains no ball of radius $\rho$, then we have $$\Vol\big(C + \Ball(\alpha)\big) \leq 2(n\rho + \alpha).$$
\end{lemma}

\begin{proof}
By the theorem of John \cite{John48} (see also Theorem~3.1 of \cite{Ball:intro-convex}), there is a unique ellipsoid contained in $C$ that has maximal Euclidean volume; let us denote this by $E(C).$
Since $C$ does not contain a ball of radius $\rho$, $E(C)$ must have some axis $u$ which has length less than $\rho$. Let us translate $C$ so that the center of $E(C)$ lies at the origin.
Again by the theorem from~John \red{(see the discussion in \cite{Ball:intro-convex} on pages 13 and 16),} we have that $C \subseteq nE(C)$.  Now consider the set $H$ of all points $v \in \R^n$ whose projection onto the $u$ direction has magnitude at most $n\rho + \alpha$. This is a ``thickened hyperplane'' which contains $C + \Ball(\alpha),$ and its Gaussian volume is given by $$\Vol(H)=\int_{-(n \rho + \alpha)}^{(n\rho + \alpha)} \varphi(x)\,dx,$$ where $\varphi(x)$ is the density function of a univariate normal distribution as defined in Section~\ref{sec:prelims}. We know that $\phi$ is bounded from above by $1$ so this integral is at most $2(n\rho + \alpha)$.  It is also easy to see that the same volume upper bound must hold upon undoing the translation of $C$ back to its original position, and the lemma is proved.
\end{proof}

\begin{lemma}\label{lem:small}
Let $C$ be a bounded convex subset of $\R^n$ that contains $\Ball(\rho)$, the origin-centered ball of radius $\rho$, for some $\blue{\rho > \alpha}$.  Then the distance between $(1- ({\alpha}/{\rho}))C$ and $\del C$ is at least $\alpha$.
\end{lemma}

\begin{proof}
This is essentially Lemma~2.2 of \cite{Kern}; for completeness we give the simple proof here.

Let $\beta=\alpha/\rho$. Let $z \in \del C$ be a point on the boundary of $C$. Since $C$ is convex and contains the origin, there exists a vector $v$ for which $v \cdot z = 1$ but for all $x \in C$ we have $v \cdot x \leq 1$ (intuitively, one can think of $v$ as defining the tangent hyperplane at $z$). Then for any $y \in (1 - \beta) C$ we have $v \cdot y \leq 1 - \beta$, which implies that $v (z - y) \geq \beta.$  Since ${{\rho v} /{\|v\|}} \in \Ball(\rho) \subseteq C$, it must be the case that $v \cdot {({\rho v} /{\|v\|})} = \rho \|v\| \leq 1$, which means that $\|v\| \leq  1/ \rho$ and thus (as $v (z - y) \geq \beta$) $\|z-y\| \geq \alpha.$
\end{proof}

\begin{lemma}\label{lem:large}
Let $C \subset \R^n $ be a convex set that satisfies $\sup_{c\in C} \|c\| \leq K$ for some $K >1$.  Then for any $0 < \beta < 1$, every point $v \in \del C + \Ball(\alpha)$ is within distance $\blue{2K\beta + \alpha}$ of a point in $(1 - \beta)C$.
\end{lemma}

\begin{proof}
We have that $v = c + y$ for some $c \in \del C$ and $y$ with $\|y\| \leq \alpha$.
While $v$ may not lie in $C$ (as $C$ might be an open set), we know for any $\eps>0$ there is
  a point $c'\in C$ and $\|c'-c\|\le \eps$.
Take such a point $c'$ with $\eps=\beta K$.
Then $(1-\beta)c'\in (1-\beta)C$ and
$$
\|(1-\beta)c'-v\|=\|(1-\beta)c'-c-y\|\le \|c'-c\|+\beta\|c'\|+\|y\|\le \beta K+\beta K+\alpha
=2\beta K+\alpha.
$$

This finishes the proof of the lemma.
\end{proof}

%% file: sbct1.bbl
\newcommand{\etalchar}[1]{$^{#1}$}
\begin{thebibliography}{KNOW14}

\bibitem[AAK{\etalchar{+}}07]{AAK+07}
Noga Alon, Alexandr Andoni, Tali Kaufman, Kevin Matulef, Ronitt Rubinfeld, and
  Ning Xie.
\newblock Testing $k$-wise and almost $k$-wise independence.
\newblock In {\em Proceedings of the 39th Annual {ACM} Symposium on Theory of
  Computing}, pages 496--505, 2007.

\bibitem[ACS10]{ACS10}
Michal Adamaszek, Artur Czumaj, and Christian Sohler.
\newblock Testing monotone continuous distributions on high-dimensional real
  cubes.
\newblock In {\em SODA}, pages 56--65, 2010.

\bibitem[ADK15]{AcharyaDK15}
Jayadev Acharya, Constantinos Daskalakis, and Gautam Kamath.
\newblock Optimal testing for properties of distributions.
\newblock In {\em Advances in Neural Information Processing Systems 28 (NIPS)},
  pages 3591--3599, 2015.

\bibitem[AHW16]{AHW16}
Noga Alon, Rani Hod, and Amit Weinstein.
\newblock On active and passive testing.
\newblock {\em Combinatorics, Probability {\&} Computing}, 25(1):1--20, 2016.

\bibitem[AKK{\etalchar{+}}05]{AKKLRtit}
N.~Alon, T.~Kaufman, M.~Krivelevich, S.~Litsyn, and D.~Ron.
\newblock {Testing Reed-Muller Codes}.
\newblock {\em IEEE Transactions on Information Theory}, 51(11):4032--4039,
  2005.

\bibitem[Bal93]{Ball:93}
K.~Ball.
\newblock {The Reverse Isoperimetric Problem for Gaussian Measure}.
\newblock {\em Discrete and Computational Geometry}, 10:411--420, 1993.

\bibitem[Bal97]{Ball:intro-convex}
Keith Ball.
\newblock An elementary introduction to modern convex geometry.
\newblock In {\em Flavors of Geometry}, pages 1--58. MSRI Publications, 1997.

\bibitem[BBBY12]{BBBY12}
Maria{-}Florina Balcan, Eric Blais, Avrim Blum, and Liu Yang.
\newblock Active property testing.
\newblock In {\em 53rd Annual {IEEE} Symposium on Foundations of Computer
  Science, {FOCS} 2012, New Brunswick, NJ, USA, October 20-23, 2012}, pages
  21--30, 2012.

\bibitem[BFRV11]{BFRV11}
Arnab Bhattacharyya, Eldar Fischer, Ronitt Rubinfeld, and Paul Valiant.
\newblock Testing monotonicity of distributions over general partial orders.
\newblock In {\em ICS}, pages 239--252, 2011.

\bibitem[BKR04]{BKR:04long}
Tugkan Batu, Ravi Kumar, and Ronitt Rubinfeld.
\newblock Sublinear algorithms for testing monotone and unimodal distributions.
\newblock In {\em {Proceedings of the 36th Symposium on Theory of Computing}},
  pages 381--390, 2004.

\bibitem[BKS{\etalchar{+}}10]{BKSSZ10}
Arnab Bhattacharyya, Swastik Kopparty, Grant Schoenebeck, Madhu Sudan, and
  David Zuckerman.
\newblock Optimal testing of reed-muller codes.
\newblock In {\em 51th Annual {IEEE} Symposium on Foundations of Computer
  Science, {FOCS} 2010}, pages 488--497, 2010.

\bibitem[Bla09]{Blaisstoc09}
Eric Blais.
\newblock Testing juntas nearly optimally.
\newblock In {\em Proc.\ 41st Annual ACM Symposium on Theory of Computing
  (STOC)}, pages 151--158, 2009.

\bibitem[BLR93]{BLR93}
M.~Blum, M.~Luby, and R.~Rubinfeld.
\newblock Self-testing/correcting with applications to numerical problems.
\newblock {\em Journal of Computer and System Sciences}, 47:549--595, 1993.
\newblock Earlier version in STOC'90.

\bibitem[BMR16a]{BMR16fsttcs}
Piotr Berman, Meiram Murzabulatov, and Sofya Raskhodnikova.
\newblock The power and limitations of uniform samples in testing properties of
  figures.
\newblock In {\em 36th {IARCS} Annual Conference on Foundations of Software
  Technology and Theoretical Computer Science, {FSTTCS} 2016, December 13-15,
  2016, Chennai, India}, pages 45:1--45:14, 2016.

\bibitem[BMR16b]{BMR16socg}
Piotr Berman, Meiram Murzabulatov, and Sofya Raskhodnikova.
\newblock Testing convexity of figures under the uniform distribution.
\newblock In {\em 32nd International Symposium on Computational Geometry, SoCG
  2016, June 14-18, 2016, Boston, MA, {USA}}, pages 17:1--17:15, 2016.

\bibitem[BMR16c]{BMR16icalp}
Piotr Berman, Meiram Murzabulatov, and Sofya Raskhodnikova.
\newblock Tolerant testers of image properties.
\newblock In {\em 43rd International Colloquium on Automata, Languages, and
  Programming, {ICALP} 2016, July 11-15, 2016, Rome, Italy}, pages 90:1--90:14,
  2016.

\bibitem[BY16]{BlaisYoshida16}
Eric Blais and Yuichi Yoshida.
\newblock A characterization of constant-sample testable properties.
\newblock {\em CoRR}, abs/1612.06016, 2016.

\bibitem[CS01]{CS:01}
Artur Czumaj and Christian Sohler.
\newblock Property testing with geometric queries.
\newblock In {\em Algorithms - {ESA} 2001, 9th Annual European Symposium},
  pages 266--277, 2001.

\bibitem[CSZ00]{CSZ:00}
Artur Czumaj, Christian Sohler, and Martin Ziegler.
\newblock Property testing in computational geometry.
\newblock In {\em Algorithms - {ESA} 2000, 8th Annual European Symposium},
  pages 155--166, 2000.

\bibitem[GGL{\etalchar{+}}00]{GGLRS}
O.~Goldreich, S.~Goldwasser, E.~Lehman, D.~Ron, and A.~Samordinsky.
\newblock Testing monotonicity.
\newblock {\em Combinatorica}, 20(3):301--337, 2000.

\bibitem[GGR98]{GGR98}
O.~Goldreich, S.~Goldwasser, and D.~Ron.
\newblock Property testing and its connection to learning and approximation.
\newblock {\em Journal of the ACM}, 45:653--750, 1998.

\bibitem[GR16]{GoldreichRon16}
Oded Goldreich and Dana Ron.
\newblock On sample-based testers.
\newblock {\em {TOCT}}, 8(2):7:1--7:54, 2016.

\bibitem[GS06]{GoldreichSudan06}
Oded Goldreich and Madhu Sudan.
\newblock Locally testable codes and pcps of almost-linear length.
\newblock {\em J. {ACM}}, 53(4):558--655, 2006.

\bibitem[GW93]{convex-geometry}
P.M. Gruber and J.M. Wills, editors.
\newblock {\em Handbook of convex geometry, Volume A}.
\newblock Elsevier, New York, 1993.

\bibitem[Joh48]{John48}
Fritz John.
\newblock Extremum problems with inequalities as subsidiary conditions.
\newblock In {\em Studies and essays presented to R. Courant on his 60th
  birthday}, pages 187--204. Interscience, New York, 1948.

\bibitem[Joh01]{Johnstone01}
Iain~M. Johnstone.
\newblock Chi-square oracle inequalities.
\newblock In {\em State of the art in probability and statistics}, pages
  399--418. Institute of Mathematical Statistics, 2001.

\bibitem[Ker92]{Kern}
W.~Kern.
\newblock Learning convex bodies under uniform distribution.
\newblock {\em Information Processing Letters}, pages 35--39, 1992.

\bibitem[KMS15]{KMS15}
Subhash Khot, Dor Minzer, and Muli Safra.
\newblock On monotonicity testing and boolean isoperimetric type theorems.
\newblock To appear in FOCS, 2015.

\bibitem[KNOW14]{KNOW:14}
Pravesh Kothari, Amir Nayyeri, Ryan O'Donnell, and Chenggang Wu.
\newblock Testing surface area.
\newblock In {\em Proceedings of the Twenty-Fifth Annual {ACM-SIAM} Symposium
  on Discrete Algorithms (SODA)}, pages 1204--1214, 2014.

\bibitem[KOS07]{KOS:07}
A.~Klivans, R.~O'Donnell, and R.~Servedio.
\newblock Agnostically learning convex sets via perimeter.
\newblock manuscript, 2007.

\bibitem[KR00]{KR00}
M.~Kearns and D.~Ron.
\newblock Testing problems with sub-learning sample complexity.
\newblock {\em Journal of Computer and System Sciences}, 61:428--456, 2000.

\bibitem[KS08]{KaufmanSudan08}
Tali Kaufman and Madhu Sudan.
\newblock Algebraic property testing: the role of invariance.
\newblock In {\em Proceedings of the 40th Annual {ACM} Symposium on Theory of
  Computing, Victoria, British Columbia, Canada, May 17-20, 2008}, pages
  403--412, 2008.

\bibitem[MORS10]{MORS:10}
K.~Matulef, R.~O'Donnell, R.~Rubinfeld, and R.~Servedio.
\newblock Testing halfspaces.
\newblock {\em SIAM J. on Comput.}, 39(5):2004--2047, 2010.

\bibitem[Naz03]{Nazarov:03}
F.~Nazarov.
\newblock {On the maximal perimeter of a convex set in $\R^n$ with respect to a
  Gaussian measure}.
\newblock In {\em Geometric aspects of functional analysis (2001-2002)}, pages
  169--187. Lecture Notes in Math., Vol. 1807, Springer, 2003.

\bibitem[Nee14]{Neeman:14}
Joe Neeman.
\newblock Testing surface area with arbitrary accuracy.
\newblock In {\em Proceedings of the 46th Annual ACM Symposium on Theory of
  Computing}, STOC '14, pages 393--397, 2014.

\bibitem[PRS02]{PRS02}
M.~Parnas, D.~Ron, and A.~Samorodnitsky.
\newblock {Testing Basic Boolean Formulae}.
\newblock {\em SIAM J. Disc. Math.}, 16:20--46, 2002.

\bibitem[Ras03]{Raskhodnikova:03}
Sofya Raskhodnikova.
\newblock Approximate testing of visual properties.
\newblock In {\em Proceedings of RANDOM}, pages 370--381, 2003.

\bibitem[RS05]{RubinfeldServedio:05}
R.~Rubinfeld and R.~Servedio.
\newblock Testing monotone high-dimensional distributions.
\newblock In {\em Proc.\ 37th Annual ACM Symposium on Theory of Computing
  (STOC)}, pages 147--156, 2005.

\bibitem[RV05]{Vempala}
Luis Rademacher and Santosh Vempala.
\newblock Testing geometric convexity.
\newblock In {\em FSTTCS 2004: Foundations of Software Technology and
  Theoretical Computer Science: 24th International Conference, Chennai, India,
  December 16-18, 2004. Proceedings}, pages 469--480, 2005.

\bibitem[RX10]{RX10}
Ronitt Rubinfeld and Ning Xie.
\newblock Testing non-uniform \emph{k}-wise independent distributions over
  product spaces.
\newblock In {\em Automata, Languages and Programming, 37th International
  Colloquium, {ICALP} 2010, Bordeaux, France, July 6-10, 2010, Proceedings,
  Part {I}}, pages 565--581, 2010.

\bibitem[Sza06]{Szarek06}
Stanislaw~J. Szarek.
\newblock Convexity, complexity, and high dimensions.
\newblock In {\em Proceedings of the International Congress of Mathematicians,
  Madrid, Spain}, pages 1599--1621. European Mathematical Society, 2006.

\end{thebibliography}
